\documentclass[11pt]{article}
\usepackage{fullpage,color}
\usepackage{amsmath,amssymb,amsthm,mathtools,enumerate,multirow,url}

\newtheorem{theorem}{Theorem}[section]
\newtheorem{lemma}[theorem]{Lemma}
\theoremstyle{definition}
\newtheorem{definition}{Definition}[section]

\providecommand{\mm}[1]{\langle #1 \rangle}
\DeclareMathOperator*{\EE}{\mathbb{E}}
\providecommand{\bR}{\underline{R}}
\providecommand{\field}{\mathbb{F}}
\providecommand{\tcw}{\mathbb{T}}
\providecommand{\tcwq}[1][q]{\tcw(#1)}
\DeclareMathOperator{\Vol}{Vol}
\providecommand{\Vg}{V} 
\providecommand{\Vcw}{V^{\mathrm{pr}}} 
\providecommand{\Vm}{V^{\mathrm{m}}} 
\providecommand{\Vmub}{V^{\mathrm{m\ast}}} 
\providecommand{\Vc}{V^{\mathrm{cm}}}
\providecommand{\omegacw}{\omega^{\mathrm{pr}}}
\providecommand{\omegam}{\omega^{\mathrm{m}}}
\providecommand{\omegamub}{\omega^{\mathrm{m}\ast}}

\DeclareMathOperator{\Val}{Val} 
\providecommand{\Valo}{\Val'} 
\providecommand{\Int}{\mathbb{Z}}
\providecommand{\pr}{partition-restricted }
\providecommand{\classtensor}{\mathcal{C}}

\providecommand{\dist}{\mathcal{D}}
\providecommand{\sdist}{\mathcal{D}^{\mathrm{sym}}}

\DeclareMathOperator{\supp}{supp}
\providecommand{\suppz}{\supp^0}
\providecommand{\suppn}{\supp^\ast}
\providecommand{\tz}[1]{#1^0}
\providecommand{\tn}[1]{#1^\ast}
\DeclareMathOperator{\Tr}{Tr}
\providecommand{\rot}[1]{#1^{\mathsf{C}}}
\providecommand{\rott}[1]{#1^{\mathsf{C}^2}}

\providecommand{\prnd}[2]{{#1\odot #2}}
\providecommand{\Pm}{P_{\mathrm{m}}}

\newif\ifsketch 
\sketchfalse
\ifsketch
\providecommand{\highlight}[1]{{\color{blue}#1}}
\providecommand{\Proofsketch}{Proof \highlight{sketch}}
\providecommand{\todo}[1]{{\color{red} #1}}
\else
\providecommand{\Proofsketch}{Proof}
\providecommand{\todo}[1]{}
\fi
\newif\ifextendedabstract
\extendedabstractfalse

\providecommand{\qS}{\sigma}
\providecommand{\qT}{\pi}

\ifextendedabstract
\title{Fast Matrix Multiplication: \\Limitations of the Laser Method \\ (Extended Abstract)}
\else
\title{Fast Matrix Multiplication: \\Limitations of the Laser Method}
\fi
\author{
Andris Ambainis\\
University of Latvia\footnote{Research done while visiting the Institute for Advanced Study.} \\ Riga, Latvia \\
\url{andris.ambainis@lu.lv}
\and
Yuval Filmus\\
Institute for Advanced Study \\ Princeton, New Jersey, USA \\
\url{yfilmus@ias.edu}
\and Fran{\c c}ois Le Gall\\
The University of Tokyo\\ Tokyo, Japan \\
\url{legall@is.s.u-tokyo.ac.jp}}

\begin{document}

\maketitle
\thispagestyle{empty}
\setcounter{page}{0}

\begin{abstract}
 Until a few years ago, the fastest known matrix multiplication algorithm, due to Coppersmith and Winograd (1990), ran in time $O(n^{2.3755})$. Recently, a surge of activity by Stothers, Vassilevska-Williams, and Le~Gall has led to an improved algorithm running in time $O(n^{2.3729})$. These algorithms are obtained by analyzing higher and higher tensor powers of a certain identity of Coppersmith and Winograd. We show that this exact approach cannot result in an algorithm with running time $O(n^{2.3725})$, and identify a wide class of variants of this approach which cannot result in an algorithm with running time $O(n^{2.3078})$; in particular, this approach cannot prove the conjecture that for every $\epsilon > 0$, two $n\times n$ matrices can be multiplied in time $O(n^{2+\epsilon})$.

 We describe a new framework extending the original laser method, which is the method underlying the previously mentioned algorithms. Our framework accommodates the algorithms by Coppersmith and Winograd, Stothers, Vassilevska-Williams and Le~Gall. We obtain our main result by analyzing this framework. The framework is also the first to explain why taking tensor powers of the Coppersmith--Winograd identity results in faster algorithms.
\end{abstract}

\newpage

\section{Introduction}
How fast can we multiply two $n \times n$ matrices? Ever since Strassen~\cite{Strassen69} improved on the $O(n^3)$ high-school algorithm, this question has captured the imagination of computer scientists. A theory of fast algorithms for matrix multiplication has been developed. Highlights include Sch\"onhage's asymptotic sum inequality~\cite{Schonhage}, Strassen's laser method~\cite{Strassen87}, and the Coppersmith--Winograd algorithm~\cite{CoppersmithWinograd}. The algorithm by Coppersmith and Winograd had been the world champion for 20 years, until finally being improved by Stothers~\cite{Stothers} in 2010. Independently, Vassilevska-Williams~\cite{Williams} obtained a further improvement in 2012, and Le~Gall~\cite{LeGall} perfected their methods to obtain the current world champion in 2014.

The Coppersmith--Winograd algorithm relies on a certain identity which we call the \emph{Coppersmith--Winograd identity}. Using a very clever combinatorial construction and the laser method, Coppersmith and Winograd were able to extract a fast matrix multiplication algorithm whose running time is $O(n^{2.3872})$. Applying their technique recursively for the tensor square of their identity, they obtained an even faster matrix multiplication algorithm with running time $O(n^{2.3755})$. For a long time, this latter algorithm had been the state of the art.

The calculations for higher tensor powers are complicated, and yield no improvement for the tensor cube. With the advent of modern computers, however, it became possible to automate the necessary calculations, allowing Stothers to analyze the fourth tensor power and obtain an algorithm with running time $O(n^{2.3730})$. Apart from implementing the necessary computer programs, Stothers also had to generalize the original framework of Coppersmith and Winograd. Independently, Vassilevska-Williams performed the necessary calculations for the fourth and eighth tensor powers, obtaining an algorithm with running time $O(n^{2.3728642})$ for the latter. Higher tensor powers require more extensive calculations, involving the approximate solution of large optimization problems. Le~Gall came up with a faster method for solving these large optimization problems (albeit yielding slightly worse solutions), and this enabled him to perform the necessary calculations for the sixteenth and thirty-second tensor powers, obtaining algorithms with running times $O(n^{2.3728640})$ and $O(n^{2.3728639})$, respectively.

It is commonly conjectured that for every $\epsilon > 0$, there exists a matrix multiplication algorithm with running time $O(n^{2+\epsilon})$. Can taking higher and higher tensor powers of the Coppersmith--Winograd identity yield these algorithms? In this paper we answer this question in the negative. We show that taking the $2^N$th tensor power cannot yield an algorithm with running time $O(n^{2.3725})$, for \emph{any} value of $N$.
We obtain this lower bound by presenting a framework which subsumes the techniques of Coppersmith and Winograd, Stothers, Vassilevska-Williams, and Le~Gall, and is amenable to analysis. At the same time, our framework is the first to explain what is gained by taking tensor powers of the original Coppersmith--Winograd identity.

All prior work follows a very rigid framework in analyzing powers of the Coppersmith--Winograd identities. However, Coppersmith and Winograd themselves already noted that there are many degrees of freedom which are not explored by this rigid framework. One such degree of freedom is analyzing powers of the identity other than powers of $2$, and another one has to do with the exact way that the tensor square of an identity is analyzed. Our new framework subsumes not only the common rigid framework used by all prior work, but also accommodated these degrees of freedom. We are able to prove that even accounting for these degrees of freedom, taking the $N$th tensor power of the Coppersmith--Winograd identity cannot yield an algorithm with running time $O(n^{2.3078})$, for \emph{any} value of $N$. This limitation holds even for our new framework, which in some sense corresponds to analyzing all powers at once.

\medskip
\paragraph{Overview of our approach}
The Coppersmith--Winograd identity bounds the \emph{border rank} (a certain measure of complexity) of a certain \emph{tensor} (three-dimensional analog of a matrix) $\tcw$. The tensor is a sum of six non-disjoint smaller tensors. Sch\"onhage's asymptotic sum inequality allows us to obtain a matrix multiplication algorithm given a bound on the border rank of a sum of \emph{disjoint} tensors of a special kind, which includes the tensors appearing in~$\tcw$. The idea of the \emph{laser method} is to take a high tensor power of~$\tcw$ and zero out some of the variables so that the surviving smaller tensors are disjoint. Applying Sch\"onhage's asymptotic sum inequality then yields a matrix multiplication algorithm. Following this route, an algorithm with running time $O(n^{2.3872})$ is obtained.

In order to improve on this, Coppersmith and Winograd take the tensor square of $\tcw$, and rewrite it as a sum of fifteen non-disjoint smaller tensors, which result from merging in a particular way the thirty-six tensors obtained from the squaring (this ``particular way'' is one of the degrees of freedom mentioned above). At this point the earlier construction is repeated (i.e., the laser method is applied on $\tcw^{\otimes 2}$). In total, the new construction is equivalent to the following procedure: start with the original tensor $\tcw$, take a high tensor power of it, zero out some of the variables, and merge groups of remaining tensors so that the resulting merged tensors are disjoint and are of the kind that allows application of the asymptotic sum inequality. The further constructions of Stothers (on the 4th power of $\tcw$), Vassilevska-Williams (on the 8th power of $\tcw$), and Le~Gall (on the 16th and 32nd powers of $\tcw$) can all be put in this framework.

Numerical calculations show that the bound on $\omega$ obtained by considering $\tcw^{\otimes 2^\ell}$ improves as $\ell$ increases, but the root cause of this phenomenon has never been completely explained (and never quantitatively studied). Indeed, at first glance it seems that considering powers of $\tcw$ should not help at all, since the analysis of $\tcw$ proceeds by analyzing powers $\tcw^{\otimes N}$ for large $N$; how do we gain anything by analyzing instead large powers of, for instance, $\tcw^{\otimes 2}$? The improvement actually results from the fact that when defining $\tcw^{\otimes 2}$ it is possible to merge together several parts of the tensor. 
Inspired by this observation, we introduce a method to analyze tensors that we call the \emph{laser method with merging}. The crucial property is that the methods used in prior works \cite{CoppersmithWinograd,DavieStothers,LeGall,Stothers,Williams}, even accounting for the degrees of freedom mentioned above, can all be put in this framework, and thus showing limitations of the laser method with merging immediately shows the limitations of all these approaches.

The first main technical contribution of this paper is a general methodology to show, quantitatively, the limitations of the laser method with merging (we stress that these techniques are currently only tailored to proving such limitations: we do not know how to systematically and efficiently convert constructions discovered through the laser method with merging into algorithms for matrix multiplication).
A summary of our results appears in Table~\ref{tab:simulations-all} on page~\pageref{tab:simulations-all}.
The Coppersmith--Winograd identity is parameterized by an integer parameter $q \geq 1$; our method applies for all these values.
 For $q,r \geq 0$, let $\omega \leq \omegam_{q,r}$ and $\omega \leq \omegacw_{q,r}$ be the bounds on $\omega$ obtained by applying the laser method with merging and applying recursively the laser method as in prior works, respectively, to the $2^r$th tensor power of $\tcw$ with the given value of~$q$. For each $q,r$, the table gives $\omegacw_{q,r}$ and a lower bound $\omegamub_{q,r}$ on $\omegam_{q,r}$.
The bound $O(n^{2.3078})$ already mentioned corresponds to the analysis of~$\tcw$ with the choice $q = 5$, which is the value used by Stothers, Vassilevska-Williams, and Le~Gall (Coppersmith and Winograd used the value $q=6$, for which our lower bound is even better). The bound $O(n^{2.3725})$ corresponds to the analysis of $\tcw^{\otimes 16}$ with the choice $q = 5$.

The second main technical contribution of this paper is to show that the laser method with merging applied on a tensor subsumes the laser method applied on any power of it. When applied to the tensor $\tcw$, this result implies that $\omegam_{q,r} \leq \omegacw_{q,s}$ for all $s \geq r$, and so $\omegacw_{q,s} \geq \omegamub_{q,r}$ for all $s \geq r$.
Combined with our results, this implies in particular that $\omegacw_{q,s}>2.3725$ for any $s\geq 4$.
Since previous works~\cite{CoppersmithWinograd,Stothers,DavieStothers,Williams,LeGall} showed that $\omegacw_{q,s}>2.3728$ for $0\leq s\leq 4$, we conclude that
analyzing \textrm{any} power of $\tcw$ recursively as done in the prior works cannot result in an algorithm with running time $O(n^{2.3725})$.

\begin{table}
\renewcommand\arraystretch{1}
\begin{center}
\caption{Upper bounds on $\omega$ obtained by analyzing $\tcw^{\otimes 2^r}$ using the laser method (L.M.), i.e., the value $\omegacw_{q,r}$, and limits on the upper bounds on $\omega$ which can possibly be obtained by analyzing $\tcw^{\otimes 2^r}$ using the laser method with merging (L.M.M.), i.e., the value $\omegamub_{q,r}$, for several values of $r$ and $q$. Note that the recursive laser method improves as we take higher and higher powers, and so the L.M.\ rows in the table are decreasing. In contrast, the laser method with merging deteriorates, since the laser method with merging applied to some power of $\tcw$ subsumes the method applied to higher powers of $\tcw$, and so the L.M.M.\ rows are increasing.
}\label{tab:simulations-all}\vspace{3mm}
\begin{tabular}{c|c|c|c|c|c|c|c|}
&Method&$r=0$&$r=1$&$r=2$&$r=3$&$r=4$\\
\hline
\multirow{2}{*}{$q=1$}&L.M.&$3$&$2.8084$&$2.6520$&$2.6324$&$2.6312$\\
&L.M.M.&$2.2387$&$2.3075$&$2.4587$&$2.5772$&$2.6184$\\\hline
\multirow{2}{*}{$q=2$}&L.M.&$2.6986$&$2.4968$&$2.4707$&$2.4690$&$2.4689$\\
&L.M.M.&$2.2540$&$2.3181$&$2.4187$&$2.4623$&$2.4673$\\\hline
\multirow{2}{*}{$q=3$}&L.M.&$2.4740$&$2.4116$&$2.4030$&$2.4027$&$2.4027$\\
&L.M.M.&$2.2725$&$2.3203$&$2.3834$&$2.4015$&$2.4025$\\\hline
\multirow{2}{*}{$q=4$}&L.M.&$2.4142$&$2.3838$&$2.3796$&$2.3794$&$2.3794$\\
&L.M.M.&$2.2907$&$2.3262$&$2.3690$&$2.3788$&$2.3791$\\\hline
\multirow{2}{*}{$q=5$}&L.M.&$2.3935$&$2.3756$&$2.3730$&$2.3729$&$2.3729$\\
&L.M.M.&$2.3078$&$2.3349$&$2.3659$&$2.3723$&$2.3725$\\\hline
\multirow{2}{*}{$q=6$}&L.M.&$2.3872$&$2.3755$&$2.3737$&$2.3737$&$2.3737$\\
&L.M.M.&$2.3234$&$2.3448$&$2.3682$&$2.3731$&$2.3733$\\\hline
\multirow{2}{*}{$q=7$}&L.M.&$2.3875$&$2.3793$&$2.3780$&$2.3779$&$2.3779$\\
&L.M.M.&$2.3377$&$2.3550$&$2.3733$&$2.3775$&$2.3776$\\\hline
\multirow{2}{*}{$q=8$}&L.M.&$2.3909$&$2.3848$&$2.3838$&$2.3838$&$2.3838$\\
&L.M.M.&$2.3508$&$2.3651$&$2.3798$&$2.3833$&$2.3834$\\\hline
\end{tabular}
\end{center}
\end{table}

\medskip

Finally, we mention that our methodology to show the limitations of the laser method with merging is related to a proof technique that appeared in a completely different approach to fast matrix multiplication developed by Cohn, Kleinberg, Szegedy and Umans \cite{CKSU,CohnUmans03,CohnUmans13}. More precisely, the combinatorial objects used in a simpler construction by Coppersmith and Winograd (given in~\cite{CoppersmithWinograd}, and corresponding to an algorithm with complexity $O(n^{2.404})$) have been studied in \cite{CKSU} under the name ``uniquely solvable puzzles", showing that this part of their construction is optimal. This argument actually implies that, in the framework of the laser method, their whole construction is optimal. We are able to give analogous bounds for the laser method with merging using similar but significantly more complicated ideas.
\ifextendedabstract
\else
\paragraph{Paper organization} Section~\ref{sec:summary} contains a longer account of our results and techniques. The main body of the paper begins with Section~\ref{sec:background}, which describes the theory of fast matrix multiplication up to and including the recent work of Stothers, Vassilevska-Williams, and Le~Gall. The laser method with merging is described in Section~\ref{sec:merging}, in which we also explain how the algorithms of Coppersmith and Winograd, Stothers, Vassilevska-Williams, and Le~Gall fit in this framework. Our main result, giving limitations of this method, appears in Section~\ref{sec:upper-bound}. The most general form of our framework and lower bound is described in Section~\ref{sec:generalizations}. We close the paper in Section~\ref{sec:discussion} by discussing our results and their implications.

\paragraph{Acknowledgements}
This material is based upon work supported by the National Science Foundation under agreement No.~DMS-1128155. Any opinions, findings and conclusions or recommendations expressed in this material are those of the authors, and do not necessarily reflect the views of the National Science Foundation.

We thank Edinah Gnang and Avi Wigderson for helpful discussions.
\fi

\section{Summary of results and techniques} \label{sec:summary}
Having briefly described our results in the introduction, we proceed to explain them in more detail. We first describe in general lines the theory of fast matrix multiplication algorithms up to the present, in Sections~\ref{sec:summary-asi}--\ref{sec:summary-rec}. We describe our new method, the laser method with merging, in Section~\ref{sec:summary-merging}, and our results in Section~\ref{sec:summary-results}.

The computation model we consider is the standard algebraic complexity model, in which a program is a list of arithmetic instructions of arity $2$. A program for multiplying two matrices has variables initialized to the entries of the input matrices, and variables designated as outputs. At the end of the program, the output variables should contain the entries of the product of the two input matrices. The complexity of the program is the number of arithmetic instructions. The \emph{matrix multiplication constant} $\omega$ is the minimal number such that for any $\epsilon > 0$, two $n\times n$ matrices can be multiplied in complexity $O(n^{\omega+\epsilon})$. Although it is conjectured that $\omega = 2$, it is not expected that the complexity is $O(n^2)$ (see for example \cite{Raz}), and this is the reason the $\epsilon$ is included.
\subsection{Fast matrix multiplication} \label{sec:summary-asi}
The first fast matrix multiplication algorithm was developed by Strassen~\cite{Strassen69}, who showed how to multiply two $2\times 2$ matrices using only $7$ scalar multiplications, implying the bound $\omega \leq \log_2 7$. He showed how to express his algorithm succinctly using the language of \emph{tensors}, a three-dimensional analog of matrices:
\begin{align*}
 \sum_{i,j,k=1}^2 x_{ij} y_{jk} z_{ki} =&
(x_{11} + x_{22})(y_{11} + y_{22})(z_{11} + z_{22})+(x_{21} + x_{22})y_{11}(z_{21} - z_{22})\\ \vspace{-3mm}
&\!\!+ x_{11} (y_{12} - y_{22})(z_{12} \!+\! z_{22}) + x_{22} (y_{21} - y_{11})(z_{11} \!+\! z_{21})+(x_{11} \!+\! x_{12}) y_{22}(-z_{11} \!+\! z_{12})\\
&\!\!+(x_{21} - x_{11}) (y_{11} + y_{12}) z_{22}+(x_{12} - x_{22}) (y_{21} + y_{22}) z_{11}.
\end{align*}
We can think of this expression as a formal trilinear form in the formal variables $x_{ij},y_{jk},z_{ki}$.
On the left we find the matrix multiplication tensor $\mm{2,2,2} = \sum_{i,j,k=1}^2 x_{ij} y_{jk} z_{ki}$ which represents the product of two $2\times 2$ matrices (indeed,
by replacing the $x$-variables by the entries of the first matrix and the $y$-variables by the entries of the second matrix, the coefficient of $z_{ki}$ in the above expression represents the entry in the $i$th row and the
$k$th column of the matrix product of these two matrices).
More generally, the $n\times m\times p$ matrix multiplication tensor $\mm{n,m,p}$ is defined as
\[
 \mm{n,m,p} = \sum_{i=1}^n \sum_{j=1}^m \sum_{k=1}^p x_{ij} y_{jk} z_{ki}.
\]
On the right of Strassen's identity we have seven \emph{rank one tensors}, which are tensors of the form $(\sum_{i'} \alpha_{i'} x_{i'})\allowbreak (\sum_{j'} \beta_{j'} y_{j'})\allowbreak (\sum_{k'} \beta_{k'} z_{k'})$; in our case $i',j'$ and $k'$ each ranges over $\{1,2\}\times\{1,2\}$. We can express the existence of such a decomposition by saying that the \emph{rank} of $\mm{2,2,2}$, denoted $R(\mm{2,2,2})$, is at most $7$ (in fact, it is exactly $7$). More generally, the rank of a tensor $T$ is the smallest $r$ such that $T$ can be written as a sum of $r$ rank one tensors. Another important concept is the \emph{border rank} of a tensor $T$, denoted $\bR(T)$, which is the smallest $r$ such that there is a sequence of tensors of rank at most $r$ converging to $T$. In the case of $\mm{2,2,2}$ the border rank and the rank are equal, but there exist tensors whose border rank is strictly smaller than their rank.

The statement $R(\mm{n,n,n})=n^{\alpha}$ is equivalent to the existence of a basis algorithm that multiplies $n\times n$ matrices with $n^{\alpha}$ multiplications.
Given such a basis algorithm, we can iterate it, obtaining an algorithm for multiplying $n^k\times n^k$ matrices with $n^{\alpha k}=(n^k)^{\alpha}$
multiplications. Thus, the matrix multiplication exponent must satisfy $\omega\leq \alpha$. We can express this argument by the inequality
$n^\omega \leq R(\mm{n,n,n})$. A vast generalization of this idea is Sch\"onhage's \emph{asymptotic sum inequality}~\cite{Schonhage}:
\[
 \sum_{i=1}^L \Vol(\mm{n_i,m_i,p_i})^{\omega/3} \leq \bR\left(\bigoplus_{i=1}^L \mm{n_i,m_i,p_i}\right).
\]
This formula uses two pieces of notation we have to explain. First, the \emph{volume} of a matrix multiplication tensor $\mm{n,m,p}$ is $\Vol(\mm{n,m,p}) = nmp$. Second, $\bigoplus_{i=1}^L \mm{n_i,m_i,p_i}$ is the \emph{direct sum} of the tensors $\mm{n_i,m_i,p_i}$. This direct sum is obtained by writing the tensors $\mm{n_i,m_i,p_i}$ using disjoint formal variables, and taking the (regular) sum. For example, $\mm{1,1,1} \oplus \mm{1,1,1} = x_1y_1z_1 + x_2y_2z_2$. How do we choose the formal variables? This does not really matter, and we identify two tensors that differ only in the names of the formal variables (we only allow renaming the $x$-variables separately, the $y$-variables separately, and the $z$-variables separately); we call two such tensors \emph{equivalent}.

As an example, Sch\"onhage showed that $\bR(\mm{4,1,4} \oplus \mm{1,9,1}) \leq 17$ (this is actually an example where the rank is strictly larger than the border rank). Applying the asymptotic sum inequality, we deduce that $16^{\omega/3} + 9^{\omega/3} \leq 17$. The left-hand side is an increasing function of $\omega$, so this inequality is equivalent to $\omega \leq \rho$ for the unique solution of $16^{\rho/3} + 9^{\rho/3} = 17$. Solving numerically for $\rho$, we obtain the bound $\omega  < 2.55$.

\subsection{Laser method} \label{sec:summary-cw}
Strassen's laser method~\cite{Strassen87} is a generalization of the asymptotic sum inequality for analyzing the sum of non-disjoint tensors.
It has been used by Coppersmith and Winograd~\cite{CoppersmithWinograd} in a particularly efficacious way.
Specifically, for an integer parameter $q$, Coppersmith and Winograd consider the following tensor:
\begin{align*}
 \tcw &= \sum_{i=1}^q \left(x_0^{[0]} y_i^{[1]} z_i^{[1]} + x_i^{[1]} y_0^{[0]} z_i^{[1]} + x_i^{[1]} y_i^{[1]} z_0^{[0]}\right) +
x_0^{[0]} y_0^{[0]} z_{q+1}^{[2]} + x_0^{[0]} y_{q+1}^{[2]} z_0^{[0]} + x_{q+1}^{[2]} y_0^{[0]} z_0^{[0]} \\ &=
 \mm{1,1,q}^{[0,1,1]} + \mm{q,1,1}^{[1,0,1]} + \mm{1,q,1}^{[1,1,0]} +
 \mm{1,1,1}^{[0,0,2]} + \mm{1,1,1}^{[0,2,0]} + \mm{1,1,1}^{[2,0,0]}.
\end{align*}

The first expression is a tensor over the $x$-variables $X = \{x_0^{[0]},x_1^{[1]},\ldots,x_q^{[1]},x_{q+1}^{[2]}\}$ and similar $y$-variables and $z$-variables. The second expression is a succinct way of describing $\tcw$, namely as a \emph{partitioned tensor}. We partition the $x$-variables into three groups: $X_0 = \{x_0^{[0]}\}$, $X_1 = \{x_1^{[1]},\ldots,x_q^{[1]}\}$, $X_2=\{x_{q+1}^{[2]}\}$, and partition the $y$-variables and $z$ variables similarly. Each of the six \emph{constituent tensors} appearing on the second line depends on a single group of $x$-variables, a single group of $y$-variables, and a single group of $z$-variables, as described by their \emph{annotations} appearing as superscripts. The constituent tensors themselves are all equivalent to matrix multiplication tensors. The notations $\mm{1,1,q}^{[0,1,1]},\mm{q,1,1}^{[1,0,1]}$ tell us that the two constituent tensors share $z$-variables but have disjoint $x$-variables and $y$-variables.

Coppersmith and Winograd showed that $\bR(\tcw) \leq q+2$ by giving an explicit sequence of tensors of rank at most $q+2$ converging to $\tcw$. This is known as the \emph{Coppersmith--Winograd identity}.

The idea of the laser method is to take a high \emph{tensor power} of $\tcw$, zero some groups of variables in order to obtain a sum of disjoint tensors, and then apply the asymptotic sum inequality. But first, we need to explain the concept of \emph{tensor product}, whose iteration gives the tensor power.
Suppose
\[
T = \sum_{i \in X} \sum_{j \in Y} \sum_{k \in Z} T_{i,j,k} x_i y_j z_k \:\:\:\:\textrm{ and }\:\:\:\:
T' = \sum_{i \in X'} \sum_{j \in Y'} \sum_{k \in Z'} T'_{i,j,k} x_i y_j z_k
\]
are two tensors.
Their tensor product $T\otimes T'$ is a tensor with $x,y,z$-variables $X\times X', Y\times Y', Z\times Z'$ given by
\[
T\otimes T' = \sum_{(i,i') \in X \times X'} \sum_{(j,j') \in Y \times Y'} \sum_{(k,k') \in Z \times Z'} T_{i,j,k} T'_{i',j',k'} x_{i,i'} y_{j,j'} z_{k,k'}.
\]
This operation corresponds to the Kronecker product of matrices. It is not hard to check that $R(T\otimes T') \leq R(T)R(T')$ and $\bR(T\otimes T') \leq \bR(T)\bR(T')$. Also, the tensor product of matrix multiplication tensors $T,T'$ is another matrix multiplication tensor satisfying $\Vol(T\otimes T') = \Vol(T)\Vol(T')$; more explicitly, $\mm{n,m,p} \otimes \mm{n',m',p'} = \mm{nn',mm',pp'}$.

When we take a high tensor power $\tcw^{\otimes N}$, we get a partitioned tensor over $X^N,Y^N,Z^N$ having~$6^N$ constituent tensors. The $x$-variables $X^N$ are partitioned into $3^N$ parts indexed by $\{0,1,2\}^N$ which we call $x$-indices; similarly we have $y$-indices and $z$-indices. Each constituent tensor of $\tcw^{\otimes N}$ has an associated \emph{index triple} which consists of its $x$-index, $y$-index and $z$-index. The constituent tensor with index triple $(I,J,K)$ is denoted $T^{\otimes N}_{I,J,K}$. The \emph{support} of $\tcw^{\otimes N}$, denoted $\supp(\tcw^{\otimes N})$, consists of the $6^N$ index triples.

Suppose we zero all $x$-variables except for those with $x$-indices in a set $A\subseteq\{0,1,2\}^N$, all $y$-variables except for those with $y$-indices in a set $B\subseteq\{0,1,2\}^N$, and all $z$-variables except for those with $z$-indices in a set $C\subseteq\{0,1,2\}^N$. The resulting tensor is
\[ \sum_{(I,J,K) \in \supp(\tcw^{\otimes N}) \cap (A\times B\times C)} T^{\otimes N}_{I,J,K}. \]
Suppose that all the summands are over disjoint variables, that is for any two different summands $T^{\otimes N}_{I_1,J_1,K_1},T^{\otimes N}_{I_2,J_2,K_2}$ we have $I_1\neq I_2,J_1\neq J_2,K_1\neq K_2$. In this case, since $\bR(\tcw^{\otimes N}) \leq (q+2)^N$, we can apply the asymptotic sum inequality to conclude that
\[ \sum_{(I,J,K) \in \supp(\tcw^{\otimes N}) \cap (A\times B\times C)} \Vol(T^{\otimes N}_{I,J,K})^{\omega/3} \leq (q+2)^N. \]

In order to analyze the construction, Coppersmith and Winograd implicitly consider the quantity $\Vcw_{\rho,N}(\tcw)$ which is the maximum of the expression $\sum_{(I,J,K) \in \supp(T^{\otimes N}) \cap (A\times B\times C)} \Vol(T^{\otimes N}_{I,J,K})^{\rho/3}$ over all $A,B,C$ which result in disjoint summands. The asymptotic sum inequality then states that $\Vcw_{\omega,N}(\tcw) \leq (q+2)^N$. It is natural to define the limit $\Vcw_\rho(\tcw) = \lim_{N\to\infty} \Vcw_{\rho,N}(\tcw)^{1/N}$ (it turns out that the limit exists), and then the asymptotic sum inequality states that $\Vcw_\omega(\tcw) \leq q+2$.

Coppersmith and Winograd were able to compute $\Vcw_\rho(\tcw)$ explicitly:
\begin{equation}\label{eq:CW}
\log_2 \Vcw_\rho(\tcw) = \max_{0 \leq \alpha \leq 1} H(\tfrac{2-\alpha}{3}, \tfrac{2\alpha}{3}, \tfrac{1-\alpha}{3}) + \tfrac{1}{3} \rho \alpha \log_2 q,
\end{equation}
where $H(\cdot)$ is the entropy function.
In fact, it is not hard to find the optimal $\alpha$ given $q$ and $\rho$.
Choosing $q = 6$, Coppersmith and Winograd calculate the value of $\rho$ which satisfies $\Vcw_\rho(\tcw) = q+2$ and deduce that $\omega \leq \rho$, obtaining the bound $\omega < 2.3872$.

Coppersmith and Winograd in fact only proved the lower bound on $\log_2 \Vcw_\rho(\tcw)$. The easier upper bound on $\log_2 \Vcw_\rho(\tcw)$ appears implicitly in the work of Cohn~et~al.~\cite{CKSU}. We sketch the proof of the upper bound since it illustrates the ideas behind our main result; this proof sketch (comprising the rest of this subsection) can be skipped on first reading.

The idea of the upper bound is simple. Given $N$, we will upper bound $\Vcw_{\rho,N}(\tcw)$ as follows. Consider any $A,B,C$ for which $\supp(\tcw^{\otimes N}) \cap (A\times B\times C)$ corresponds to disjoint tensors. For any $(I,J,K) \in \supp(\tcw^{\otimes N}) \cap (A\times B\times C)$, its \emph{source distribution} is the number of times each of the basic six tensors was used to generate $(I,J,K)$; this is a list of six non-negative integers summing to $N$. A key observation is that there are only $O(N^5)$ distinct source distributions. We upper bound the contribution of any given source distribution $\sigma$ to the sum
\[
\sum_{(I,J,K) \in \supp(\tcw^{\otimes N}) \cap (A\times B\times C)} \Vol(T^{\otimes N}_{I,J,K})^{\rho/3}
\]
as follows. First, $\Vol(T^{\otimes N}_{I,J,K})^{\rho/3}$ depends only on $\sigma$: $\log_2 \Vol(T^{\otimes N}_{I,J,K})^{\rho/3} = (\rho/3) N \EE_{r \sim \frac{\sigma}{N}} \log_2 \Vol(T_r)$, where the $\tcw_r$ are the six constituent tensors of $\tcw$. Second, since all $x$-indices appearing in the sum are distinct, the number of summands is at most the number of distinct $x$-indices which appear in index triples of type $\sigma$. There are at most $2^{NH(\sigma_1/N)}$ of these, where $\sigma_1$ is the projection of $\sigma$ on the first index. Considering also $y$-indices and $z$-indices, we obtain the bound
\begin{align*}
 &\log_2 \left(\sum_{(I,J,K) \in \supp(\tcw^{\otimes N}) \cap (A\times B\times C)} \Vol(T^{\otimes N}_{I,J,K})^{\rho/3}\right) \\
&\leq \log_2 \left(\sum_\sigma 2^{N\max(H(\sigma_1/N),H(\sigma_2/N),H(\sigma_3/N)) + \frac{\rho}{3} N \EE_{r \sim \frac{\sigma}{N}} \log_2 \Vol(\tcw_r)}\right) \\
& \leq\log_2 \left(O(N^5)\right) + \max_\sigma \left[N\max(H(\sigma_1/N),H(\sigma_2/N),H(\sigma_3/N)) + \frac{\rho}{3} N \EE_{r \sim \sigma/N} \log_2 \Vol(\tcw_r)\right].
\end{align*}
This is an upper bound on $\log_2 \Vcw_{\rho,N}(\tcw)$. Taking the limit $N\to\infty$, we obtain
\[
 \log_2 \Vcw_\rho(\tcw) \leq \max_\sigma \left[\max(H(\sigma_1),H(\sigma_2),H(\sigma_3)) + \frac{\rho}{3}  \EE_{r \sim \sigma} \log_2 \Vol(\tcw_r)\right],
\]
where this time $\sigma$ ranges over all probability distributions over $\supp(\tcw)$. This upper bound can be massaged to obtain the right-hand side of Equation (\ref{eq:CW}).

\subsection{Recursive laser method} \label{sec:summary-rec}
Coppersmith and Winograd went on to prove an even better bound by considering the square of their original identity, which shows that $\bR(\tcw^{\otimes2}) \leq (q+2)^2$. They consider $\tcw^{\otimes2}$ as a partitioned tensor, but instead of using $3^2$ parts for each type of variables, they collapse those into $5$ different parts: $X^2_i = \sum_{i_1+i_2=i} X_{i_1} \otimes X_{i_2}$ for $i\in\{0,1,2,3,4\}$ (where $X_0,X_1,X_2$ is the original partition of the $x$-variables), and similarly for the $y$-variables and $z$-variables. For example, $X^2_2$ consists of the union of $X_0 \otimes X_2,X_1\otimes X_1,X_2\otimes X_0$. According to the rules of partitioned tensors, we now have to partition $\tcw^{\otimes 2}$ into constituent tensors which only use one group each of $x$-variables, $y$-variables and $z$-variables. When we do this we obtain $15$ constituent tensors (rather than $6^2$):
\begin{itemize}
 \item The constituent tensor with index triple $(0,0,4)$ is $\mm{1,1,1}^{[0,0,2]} \otimes \mm{1,1,1}^{[0,0,2]}$, which is a matrix multiplication tensor $\mm{1,1,1}$. The index triples $(0,4,0),(4,0,0)$ can be analyzed similarly.
 \item The constituent tensor with index triple $(0,1,3)$ is $\mm{1,1,q}^{[0,1,1]} \otimes \mm{1,1,1}^{[0,0,2]} + \mm{1,1,1}^{[0,0,2]} \otimes \mm{1,1,q}^{[0,1,1]}$, which is equivalent to a \emph{single} matrix multiplication tensor $\mm{1,1,2q}$ (essentially since all four correspond to inner products whose ``result'' is in $x^{[0]}_0$). The index triples $(0,3,1),(1,0,3),(1,3,0),(3,0,1),(3,1,0)$ can be analyzed similarly.
 \item The constituent tensor with index triple $(0,2,2)$ is $\mm{1,1,1}^{[0,2,0]} \otimes \mm{1,1,1}^{[0,0,2]} + \mm{1,1,1}^{[0,0,2]} \otimes \mm{1,1,1}^{[0,2,0]}
+ \mm{1,1,q}^{[0,1,1]} \otimes \mm{1,1,q}^{[0,1,1]}$, which is equivalent to the single matrix multiplication tensor $\mm{1,1,q^2+2}$. The index triples $(2,0,2),(2,2,0)$ can be analyzed similarly.
 \item The constituent tensor with index triple $(1,1,2)$ is $\mm{1,q,1}^{[1,1,0]} \otimes \mm{1,1,1}^{[0,0,2]} + \mm{1,1,1}^{[0,0,2]} \otimes \mm{1,q,1}^{[1,1,0]} +
\mm{q,1,1}^{[1,0,1]} \otimes \mm{1,1,q}^{[0,1,1]} + \mm{1,1,q}^{[0,1,1]} \otimes \mm{q,1,1}^{[1,0,1]}$. This tensor is \emph{not} equivalent to a matrix multiplication tensor. A similar problem occurs for the index triples $(1,2,1),(2,1,1)$.
\end{itemize}

The basic idea behind the analysis of $\tcw^{\otimes2}$ is to apply the same sort of analysis we used for $\tcw$. The problem is that now we have three constituent tensors which are not matrix multiplication tensors. Coppersmith and Winograd noticed that $\tcw^{\otimes 2}_{1,1,2}$ and the other problematic tensors can be analyzed by applying the same sort of analysis once again. Define $\Val_\rho(\tcw^{\otimes 2}_{1,1,2})$ in the same way that we defined $\Vcw_\rho(\tcw)$ before\footnote{Since $\tcw^{\otimes 2}_{1,1,2}$ is not symmetric the actual definition is slightly different, and involves symmetrizing this tensor. For the sake of exposition we ignore these details here.}. The value $\Val_\rho(\tcw^{\otimes 2}_{1,1,2})$ is a number $V$ such that if we take the $N$th tensor power of $\tcw^{\otimes 2}_{1,1,2}$ and zero some variables appropriately, we get a sum of disjoint matrix multiplication tensors $\sum_s t_s$ such that $\sum_s \Vol(t_s)^{\rho/3} \approx V^N$. It can therefore be used as a replacement for the volume in an application of the asymptotic sum inequality for analyzing the tensor $\tcw^{\otimes 2}$ itself.

Coppersmith and Winograd diligently calculate $\Val_\rho(\tcw^{\otimes 2}_{1,1,2}) = 4^{1/3} q^\rho (2 + q^{3\rho})^{1/3}$, and use this value to calculate $\Vcw_\rho(\tcw^{\otimes 2})$; this time the explicit formula is too cumbersome to write concisely. Choosing $q = 6$, they are able to prove the bound $\omega < 2.3755$.

Stothers~\cite{Stothers} and Vassilevska-Williams~\cite{Williams} were the first to explain how to generalize this analysis to higher powers of the original identity, Stothers analyzing the fourth power, and Vassilevska-Williams the fourth and eighth powers. Le~Gall~\cite{LeGall} managed to analyze even higher powers: the sixteenth and thirty-second. The bound obtained by analyzing the fourth power is $\omega < 2.3730$. The analysis of higher powers results in better bounds, as shown in Table \ref{tab:simulations-all}, but those differ by less than $10^{-3}$ from the bound $2.3730$.

\subsection{Laser method with merging} \label{sec:summary-merging}

Why does analyzing $\tcw^{\otimes 2}$ result in better bounds than analyzing $\tcw$? 
The analysis of $\tcw$ proceeds by taking the $N$th tensor power, zeroing some variables, and comparing the total value of the resulting expression to $(q+2)^N$. In contrast, the analysis of $\tcw^{\otimes 2}$ proceeds by taking the $N$th tensor power of $\tcw^{\otimes 2}$, which is also the $2N$th tensor power of $\tcw$, zeroing some variables, and comparing the total value of the resulting expression to $(q+2)^{2N}$. Where is the gain?

In this paper we point out the core reason why the analysis of $\tcw^{\otimes 2}$ gains over the analysis of~$\tcw$, and evaluate it quantitatively: the gain lies in the fact that when describing the constituent tensors of $\tcw^{\otimes 2}$, we merge several non-disjoint matrix multiplication tensors to larger matrix multiplication tensors. For example, $\tcw^{\otimes 2}_{0,1,3}$ results from merging the two matrix multiplication tensors $\mm{1,1,q}^{[0,1,1]} \otimes \mm{1,1,1}^{[0,0,2]}, \mm{1,1,1}^{[0,0,2]} \otimes \mm{1,1,q}^{[0,1,1]}$ to a bigger one of shape $\mm{1,1,2q}$.

Accordingly, we define a generalization of the laser method which allows such merging. A single application of this method subsumes the analysis of all powers of the Coppersmith--Winograd identity. Recall that we defined $\Vcw_{\rho,N}(\tcw)$ to be the maximum value of
\[
\sum_{(I,J,K) \in \supp(\tcw^{\otimes N}) \cap (A\times B\times C)} \Vol(T^{\otimes N}_{I,J,K})^{\rho/3}
\]
over all choices of $A,B,C$ that result in disjoint tensors. The quantity $\Vm_{\rho,N}(\tcw)$ is defined in a similar fashion. First, we choose $A\subseteq X^N,B\subseteq Y^N,C\subseteq Z^N$ and zero all variables not in $A,B,C$. The result is a bunch of matrix multiplication tensors which we call the \emph{surviving tensors}. There follows a \emph{merging stage}: if the sum of a set of surviving tensors is equivalent to a matrix multiplication tensor, then we allow the set to be replaced by a single matrix multiplication tensor. The result of this stage is a set $\{ t_s \}$ of matrix multiplication tensors, each of which is a sum of surviving tensors. The \emph{merging value} $\Vm_{\rho,N}(\tcw)$ is defined as the maximum of $\sum_s \Vol(t_s)^{\rho/3}$ over all choices of $A,B,C$ and all mergings which result in a set of tensors $\{ t_s \}$ which have disjoint $x$-variables, $y$-variables, and $z$-variables. Mimicking the earlier definition, we define $\Vm_\rho(\tcw) = \lim_{N\to\infty} \Vm_{\rho,N}(\tcw)^{1/N}$, where again the limit always exists. A generalization of the asymptotic sum inequality shows that $\Vm_\omega(\tcw) \leq \bR(\tcw)$, and so this method allows us to prove upper bounds on $\omega$ once we have lower bounds on $\Vm_\omega(\tcw)$.

While we defined above the merging value $\Vm_\rho$ only for the Coppersmith--Winograd tensor $\tcw$, the same definition works for any other partitioned tensor whose constituent tensors are matrix multiplication tensors. A more complicated definition exists for the more general case, in which some constituent tensors are not matrix multiplication tensors, but instead are supplied with a \emph{value} (as in the recursive laser method). This definition only allows merging of matrix multiplication tensors, and is actually crucial for our analyses.

\subsection{Our results} \label{sec:summary-results}
We show that $\Vm_\rho(\tcw) \geq \Vcw_\rho(\tcw^{\otimes N})^{1/N}$ for \emph{any} $N$, where the latter is calculated along the lines of the work by Coppersmith and Winograd, Stothers, Vassilevska-Williams, and Le~Gall, but allowing more degrees of freedom than those used in current work. First, when going from $\tcw^{\otimes N}$ to $\tcw^{\otimes 2N}$, in current work the parts are always folded by putting $X^N_i \times X^N_j$ in part $X^{2N}_{i+j}$; our upper bound on~$\Vcw_\rho$ is oblivious to this choice, and thus allows arbitrary repartitioning
(this is a degree of freedom already mentioned in the original paper of Coppersmith and Winograd).
Second, current methods calculate successive values of $\tcw,\tcw^{\otimes 2},\tcw^{\otimes 4},\ldots$, each time squaring the preceding tensor. Our method also allows sequences such as $\tcw,\tcw^{\otimes 2},\tcw^{\otimes 3}$, where the latter is obtained by considering the tensor product of the two former tensors. Third, our results apply by and large to tensors other than the Coppersmith--Winograd tensor, though other (promising) such examples are not currently known.

The bound $\Vm_\rho(\tcw) \geq \Vcw_\rho(\tcw^{\otimes N})^{1/N}$ implies a limit on what bounds on $\omega$ can be obtained by the recursive laser method applied to all powers of $\tcw$. Indeed, suppose that $\rho$ is the solution to $\Vm_\rho(\tcw) = q+2$. Then the solution $\alpha$ to $\Vcw_\alpha(\tcw^{\otimes N}) = (q+2)^N$ satisfies $\alpha \geq \rho$ (because both $\Vm_\rho(\tcw)$ and $\Vcw_\alpha(\tcw^{\otimes N})$ are increasing functions of $\rho$ and $\alpha$, respectively), and so the corresponding bound $\omega \leq \alpha$ is no better than $\omega \leq \rho$.

We moreover show that more generally $\Vm_\rho(\tcw^{\otimes N_1}) \geq \Vcw_\rho((\tcw^{\otimes N_1})^{\otimes N_2})^{1/N_2}$ for any positive integers $N_1$ and $N_2$. This implies a limit on what bounds on $\omega$ can be obtained by the recursive laser method applied to $\tcw^{\otimes N_1}$. In particular, an upper bound on $\Vm_\rho(\tcw^{\otimes 2^N})$ implies a limit on what bound on $\omega$ can be obtained by the recursive laser method applied to $\tcw^{\otimes 2^M}$ for all $M \geq N$.

One of the main contributions of the paper is to show general lower bounds on the merging value.
First, using ideas similar to the upper bound on $\Vcw_\rho(\tcw)$ mentioned in Section~\ref{sec:summary-cw}, but relying on significantly more complicated arguments, we obtain an upper bound on $\Vm_\rho(\tcw)$:
\[
 \log_2 \Vm_\rho(\tcw) \leq \max_{0 \leq \alpha \leq 1} H(\tfrac{2-\alpha}{3},\tfrac{2\alpha}{3},\tfrac{1-\alpha}{3}) + \tfrac{1}{3} \rho \alpha \log_2 q + \tfrac{\rho-2}{3} H(\tfrac{1-\alpha}{2},\alpha,\tfrac{1-\alpha}{2}),
\]
gaining an extra term compared to the value of $\log_2 \Vcw_\rho(\tcw)$ given in Equation (\ref{eq:CW}). We do not have a matching lower bound, and indeed we suspect that our upper bound is not tight. Using this upper bound we find that for $q=5$, the solution to $\Vm_\rho(\tcw) = q+2$ satisfies $\rho > 2.3078$. Therefore, for $q=5$, no analysis of any power of $\tcw$, even accounting for the degrees of freedom mentioned above, can yield a bound better than $\omega < 2.3078$.

We then give a similar lower bounds on the merging value of a large class of tensors, including those of the form~$\tcw^{\otimes N_1}$.
Using this upper bound, we find that the solution to $\Vm_\rho(\tcw^{\otimes 16}) = (q+2)^{16}$ satisfies $\rho > 2.3725$, which implies that no analysis of any power $\tcw^{\otimes 2^N}$ along previous lines can yield a bound better than $\omega < 2.3725$. In particular, the existing bound $\omega < 2.3729$ cannot be improved significantly by considering the 64th, 128th, 256th powers and higher powers of the Coppersmith--Winograd identity.

\medskip

Table~\ref{tab:simulations-all} on page~\pageref{tab:simulations-all} summarizes our numerical results. For each $q,r$ the table contains the solution~$\omegamub_{q,r}$, rounded down to four decimal digits,  to $\Vmub_{\omegamub_{q,r}}(\tcw^{\otimes 2^r}) = (q+2)^{2^r}$, where $\Vmub_\rho(\tcw^{\otimes 2^r})$ is the upper bound on $\Vm_\rho(\tcw^{\otimes 2^r})$ that we obtain. In particular, the best bound $\omega \leq \omegam_{q,r}$ obtainable by applying the laser method with merging to $\tcw^{\otimes 2^r}$ satisfies $\omegam_{q,r} \geq \omegamub_{q,r}$. Since $\Vm_\rho(\tcw^{\otimes 2^r}) \geq \Vcw_\rho(\tcw^{\otimes 2^s})^{2^{r-s}}$ for all $s \geq r$, we deduce that $\omegacw_{q,s} \geq \omegam_{q,r} \geq \omegamub_{q,r}$, and so the table indeed gives limits on the upper bounds on $\omega$ which can be obtained using the recursive laser method applied to powers of $\tcw$.

\ifextendedabstract
\else
\medskip

As briefly mentioned above, our upper bound on $\Vm_\rho(T)$ applies to a wide class of tensors, with one caveat. The tensor powers of $\tcw$ have a special structure which allows us to describe what kinds of mergings are possible at the merging stage. It turns out that all such mergings satisfy the following property, which we call \emph{coherence}: if $S \subseteq \supp(\tcw^{\otimes N})$ is a set of index triples of tensors whose sum is equivalent to a matrix multiplication tensor, then for all $t \in [N]$, either $I_t = 0$ for all $(I,J,K) \in S$, or $J_t = 0$ for all $(I,J,K) \in S$, or $K_t = 0$ for all $(I,J,K) \in S$. For example, $\tcw^{\otimes 2}_{0,2,2}$ results from merging tensors corresponding to the index triples $\{(00,11,11),(00,02,20),(00,20,02)\}$, and $I_1=I_2=0$ in all of them. Our upper bound applies for all tensors for which all mergings are coherent. We get the claimed general bound by modifying the definition of the merging value, requiring that all mergings be coherent (which actually happens in all current approaches based on the laser method).
\fi

\ifextendedabstract

\else
\section{Background} \label{sec:background}
\paragraph{Notation} We write $[n] = \{1,\ldots,n\}$ 
 and use the notation $\exp_2 x$ for $2^x$. All our logarithms are to base~$2$. The entropy function $H$ is given by
\[ H(p_1,\dots,p_m) = -\sum_{i=1}^m p_i \log p_i, \]
where $0 \log 0 = 0$,
for any probability distribution $\vec{p}=(p_1,\ldots,p_m)$.
It can be used to estimate multinomial coefficients:
\[
 \binom{n}{np_1,\ldots np_m} \leq \exp_2 \left(H(p_1,\ldots,p_m)n\right).
\]
The entropy function is \emph{concave}: if $\vec{q}_1,\ldots,\vec{q}_r$ are probability distributions and $c_1,\ldots,c_r \geq 0$ sum to $1$ then
\[ \sum_{i=1}^r c_i H(\vec{q}_i) \leq H\left(\sum_{i=1}^r c_i \vec{q}_i\right). \]

The rest of this section is organized as follows:
\begin{itemize}
\item Section~\ref{sec:bilinear} describes the computational model and includes basic definitions: tensors, tensor rank, border rank, and so on. We also state Sch\"onhage's asymptotic sum inequality.
\item Section~\ref{sec:value} describes the general notion of value and the corresponding generalization of the asymptotic sum inequality.
\item Section~\ref{sec:partitioned} describes partitioned tensors, a concept which forms part of the traditional description of the laser method.
\item Section~\ref{sec:laser} gives a general version of the original Coppersmith--Winograd bound on the first power of their identity. This section includes non-standard definitions attempting to capture their construction, as well as some non-standard results which abstract the Coppersmith--Winograd method. Some of these results have not appeared before, and their proofs are given in the appendix.
\item Section~\ref{sec:cw-powers} describes the recursive version of the laser method, used by Coppersmith and Winograd~\cite{CoppersmithWinograd}, Stothers~\cite{Stothers,DavieStothers}, Vassilevska-Williams~\cite{Williams}, and Le~Gall~\cite{LeGall} to obtain the best known bounds on $\omega$.
\end{itemize}

The goal of this section is to describe the recursive laser method in enough detail so that we are able to show in Section~\ref{sec:merging} that our new variant of the method (which is not recursive) subsumes all earlier work.

\subsection{Bilinear complexity} \label{sec:bilinear}
The material below can be found in Chapters 14--15 of the book Algebraic Complexity Theory~\cite{ACT}.

\paragraph{The model} In this paper is to study the complexity of matrix multiplication in the algebraic complexity model. In this model, a program for computing the product $C = AB$ of two $n\times n$ matrices is allowed to use the following instructions:
\begin{itemize}
 \item Reading the input: $t \gets a_{ij}$ or $t \gets b_{ij}$.
 \item Arithmetic: $t \gets t_1 \circ t_2$, where $\circ \in \{+,-,\times,\div\}$.
 \item Output: $c_{ij} \gets t$.
\end{itemize}
Each of these instructions has unit cost.
All computations are done over a field $\field$, whose identity for our purposes is not so important; the reader can assume that we always work over the real numbers.
A legal program is one which never divides by zero; Strassen~\cite{Strassen73} showed how to eliminate divisions at the cost of a constant blowup in size. Denote by $T(n)$ the size of the smallest program which computes the product of two $n\times n$ matrices. The \emph{exponent of matrix multiplication} is defined by
\[ \omega = \lim_{n\to\infty} T(n)^{1/n}. \]
It can be shown that the limit indeed exists. For each $\epsilon > 0$, we also have $T(n) = O_\epsilon(n^{\omega+\epsilon})$, and $\omega$ can also be defined via this property.

\paragraph{Tensors and tensor rank} Strassen~\cite{Strassen69} related $\omega$ to the tensor rank of matrix multiplication tensors, a connection we proceed to explain. The tensors we are interested in are three-dimensional equivalents of matrices. An $n\times m$ matrix~$A$ over a field $\field$ corresponds to the bilinear form
\[
\sum_{i=1}^n \sum_{j=1}^m A_{ij} x_i y_j,
\]
where the $x_i$'s and the $y_j$'s are formal variables. Its rank is the smallest integer $r$ such that the bilinear form can be written as
\[
\sum_{s=1}^t \left(\sum_{i=1}^n\alpha_{is} x_i\right) \left(\sum_{j=1}^m\beta_{js} y_i\right)
\]
for some elements $\alpha_{is}$ and $\beta_{js}$ in $\field$.

Similarly, third order tensors correspond to trilinear forms. Let $X=\{x_1,\ldots, x_n\}$,
$Y=\{y_1,\ldots, y_m\}$ and $Z=\{z_1,\ldots, z_p\}$ be three sets of formal variables.
We call the variables in $X$ the \emph{$x$-variables}, and define \emph{$y$-variables} and \emph{$z$-variables} similarly.
A \emph{tensor over $X, Y, Z$} is a trilinear form
\[
T=\sum_{i=1}^n \sum_{j=1}^m \sum_{k=1}^p T_{ijk} x_i y_j z_k,
\]
where the $T_{ijk}$ are elements in $\field$. The \emph{rank} of $T$ is the smallest integer $r$ such that this trilinear form can be written as
\[
\sum_{s=1}^t \left(\sum_{i=1}^n\alpha_{is} x_i\right) \left(\sum_{j=1}^m\beta_{js} y_i\right)\left(\sum_{k=1}^p\gamma_{ks} z_k\right)
\]
for some elements $\alpha_{is}$, $\beta_{js}$ and $\gamma_{ks}$ in $\field$.
We denote the rank of a tensor $T$ by $R(T)$. In contrast to matrix rank, tensor rank is NP-hard to compute~\cite{Hastad,HillarLim}.

The \emph{matrix multiplication tensor} $\mm{n,m,p}$ is given by
\[ T = \sum_{i=1}^n \sum_{j=1}^m \sum_{k=1}^p x_{ij} y_{jk} z_{ki}. \]
This is an $nm \times mp \times pn$ tensor which corresponds to the trilinear product $\Tr (xyz)$, where $x,y,z$ are interpreted as $n\times m,m\times p,p\times n$ matrices, correspondingly. Strassen~\cite{Strassen69} proved that
\[ \omega = \lim_{n\to\infty} R(\mm{n,n,n})^{1/n}. \]

\paragraph{Border rank and the asymptotic sum inequality} Sch\"onhage's asymptotic sum inequality~\cite{Schonhage} is a fundamental theorem which is the main vehicle used for proving upper bounds on $\omega$. In order to state it, we need two more definitions: direct sum and border rank.

For matrices $A_1,A_2$ of dimensions $n_1\times m_1,n_2\times m_2$, their direct sum $A_1 \oplus A_2$ is the $(n_1+n_2)\times(m_1+m_2)$ block-diagonal matrix having as blocks $A_1,A_2$. Similarly we can define the \emph{direct sum} of two tensors $T_1,T_2$. 

If $A_i$ is a sequence of matrices converging to a matrix $A$, then $R(A_i) \to R(A)$. The same does not necessarily hold for tensors: if $T_i$ is a sequence of tensors converging to a tensor $T$, all we are guaranteed is that $\lim_i R(T_i) \leq R(T)$.
The \emph{border rank} of a tensor $T$, denoted $\bR(T)$, is the smallest rank of a sequence of tensors converging to $T$. Equivalently, the border rank of $T$ is the smallest rank over $\field[\epsilon]$ of any tensor of the form $\epsilon^k T + \sum_{\ell=k+1}^r \epsilon^\ell T_\ell$ (the equivalence is not immediate but follows from a result of Strassen~\cite{Strassen87}, see~\cite[\S 20.6]{ACT}). We denote any tensor of the latter form by $\epsilon^k T + O(\epsilon^{k+1})$.

We can now state the asymptotic sum inequality.

\begin{theorem}[Asymptotic sum inequality]
 For every set $n_i,m_i,p_i$ ($1 \leq i \leq K$) of positive integers,
\[
 \sum_{i=1}^K (n_i m_i p_i)^{\omega/3} \leq \bR\left(\bigoplus_{i=1}^K \mm{n_i,m_i,p_i}\right).
\]
\end{theorem}

If we define the \emph{volume} of a matrix multiplication tensor $\mm{n,m,p}$ by $\Vol(\mm{n,m,p}) = nmp$, the size of the \emph{support} (set of non-zero entries) of $\mm{n,m,p}$, then we can restate the asymptotic sum inequality as follows.

\begin{theorem}[Asymptotic sum inequality (restated)]
 For every set $T_1,\ldots,T_K$ of matrix multiplication tensors,
\[
 \sum_{i=1}^K \Vol(T_i)^{\omega/3} \leq \bR\left(\bigoplus_{i=1}^K T_i\right).
\]
\end{theorem}

The work of Coppersmith and Winograd, Stothers, Vassilevska-Williams, and Le~Gall uses a generalization of the asymptotic sum inequality in which \emph{volume} is replaced by a more general parameter which applies to arbitrary tensors rather than only to matrix multiplication ones. One main difference is that this more general notion of \emph{value} depends on $\omega$. We describe this version in Section~\ref{sec:value}.

\paragraph{Isomorphism, restriction, degeneration, and equivalence of tensors} A tensor $T$ over $X,Y,Z$ is a \emph{restriction} of a tensor $T'$ over $X',Y',Z'$ if there are linear transformations $A\colon \field[X] \to \field[X']$, $B\colon \field[Y] \to \field[Y']$, $C\colon \field[Z] \to \field[Z']$ such that $T(x,y,z) = T'(Ax,By,Cz)$ as trilinear forms over $X,Y,Z$ (here $x$ is the vector of formal $x$-variables, and $y,z$ are defined similarly). It is not hard to check that $R(T) \leq R(T')$ and $\bR(T) \leq \bR(T')$. If $T$ and $T'$ are each a restriction of the other, then we say that $T$ and $T'$ are \emph{isomorphic}\footnote{The reader might wonder what is the relation between $T$ and $T'$ if there are \emph{regular} $A,B,C$ such that $T(x,y,z) = T'(Ax,By,Cz)$. This definition is less general than isomorphism, since for example all zero tensors are isomorphic, but bijections $A,B,C$ exist only if $|X|=|X'|,|Y|=|Y'|,|Z|=|Z'|$. The exact relation between the two definitions appears in~\cite[\S14.6]{ACT}.}. Isomorphic tensors have the same rank and border rank.



There is a weaker notion of restriction which implies $\bR(T) \leq \bR(T')$. We say that $T$ is a \emph{degeneration} of $T'$ if for some $k$, $\epsilon^k T + O(\epsilon^{k+1})$ is a restriction of $T'$ over the field $\field[\epsilon]$. As shown by Strassen~\cite{Strassen87}, if two tensors are each a degeneration of the other, then they are isomorphic.

Using the notions of restriction and degeneration, we can give an alternative definition of rank and border rank. Let $\mm{n} = \sum_{i=1}^n x_i y_i z_i$ be the \emph{triple product tensor}. Then $R(T) \leq r$ if and only if $T$ is a restriction of $\mm{r}$, and $\bR(T) \leq r$ if and only if $T$ is a degeneration of $\mm{r}$.

While isomorphism and degeneration are natural concepts from an algebraic viewpoint, in practice many of the constructions appearing below are combinatorial, and so the corresponding tensors satisfy stronger relations. A tensor $T$ over $X,Y,Z$ is \emph{equivalent} to a tensor $T'$ over $X',Y',Z'$, in symbols $T\approx T'$, if there exist bijections $\alpha\colon X\to X'$, $\beta\colon Y\to Y'$, $\gamma\colon Z\to Z'$ such that $T_{ijk} = T'_{\alpha(i)\beta(j)\gamma(k)}$, that is, if $T$ and $T'$ differ by a renaming of variables. We often consider tensors only up to equivalence. If $\alpha,\beta,\gamma$ are only required to be injections, then $T$ is a \emph{combinatorial restriction} of $T'$. In that case, $T$ is obtained from $T'$ by zeroing some variables and renaming the rest arbitrarily.

\paragraph{Useful operations on tensors} Two useful operations on tensors are \emph{tensor product} (corresponding to the Kronecker product of matrices) and \emph{rotation} (corresponding to transposition of matrices).

The Kronecker or tensor product of matrices $A_1 \otimes A_2$ is an $n_1n_2 \times m_1m_2$ matrix whose entries are $(A_1\otimes A_2)_{i_1i_2,j_1j_2} = (A_1)_{i_1,j_1} (A_2)_{i_2,j_2}$. The \emph{tensor product} of two tensors is defined analogously. It then follows immediately that $\mm{n_1,m_1,p_1} \otimes \mm{n_2,m_2,p_2} \approx \mm{n_1n_2,m_1m_2,p_1p_2}$. The $n$th \emph{tensor power} of a tensor $T$ is denoted by $T^{\otimes n}$. Both rank and border rank are submultiplicative: $R(T_1\otimes T_2) \leq R(T_1) R(T_2)$ and $\bR(T_1\otimes T_2) \leq \bR(T_1) \bR(T_2)$.

Matrices can be transposed. The corresponding operation for tensors is \emph{rotation}. For an $n\times m\times p$ tensor $T = \sum_{ijk} T_{ijk} x_i y_j z_k$, its rotation is the $m\times p\times n$ tensor $\rot{T} = \sum_{jki} T_{ijk} y_j z_k x_i$. Repeating the operation again, we obtain a $p\times n\times m$ tensor $\rott{T}$. All rotations of a tensor have the same rank and the same border rank.
There are several corresponding notions of symmetry, among which we choose the one most convenient for us: a tensor $T$ is \emph{symmetric} if $\rot{T} \approx T$. 

\subsection{The value of a tensor} \label{sec:value}
The asymptotic sum inequality can be generalized to tensors which are not matrix multiplication tensors. The idea is to define a notion of value generalizing that of volume.

\begin{definition} \label{def:value}
For a tensor $T$, any $\rho\in[2,3]$, and any integer $N\ge 1$,
let $\Vg_{\rho,3N}(T)$ be the maximum of $\sum_{i=1}^L (n_i m_i p_i)^{\rho/3}$ over all degenerations of $(T\otimes\rot{T}\otimes\rott{T})^{\otimes N}$ isomorphic to $\bigoplus_{i=1}^L \mm{n_i,m_i,p_i}$. The \emph{value} of $T$ is the function
\[
 \Vg_\rho(T) = \lim_{N\to\infty} \Vg_{\rho,3N}(T)^{1/3N}.
\]
\end{definition}

When $T$ is symmetric, like the Coppersmith--Winograd tensor described below, we can do away with $T\otimes\rot{T}\otimes\rott{T}$, considering instead $T^{\otimes N}$. The more general definition is needed only for non-symmetric tensors, which do, however, come up in the analysis.

Stothers~\cite{Stothers,DavieStothers} showed that the limit in the definition of $\Vg_\rho(T)$ always exists. Furthermore, he showed that the definition of $\Vg_{\rho,N}(T)$ is unchanged if we require all dimension triples $(n_i,m_i,p_i)$ to be the same. He also proved the following properties of the value.

\begin{lemma}[\cite{DavieStothers}] \label{lem:value-properties}
 For any $\rho\in[2,3]$ the following hold:
\begin{enumerate}
 \item If $T = \mm{n,m,p}$ then $\Vg_\rho(T) = \Vol(T)^{\rho/3}$.
 \item For any $T_1,T_2$ we have $\Vg_\rho(T_1 \oplus T_2) \geq \Vg_\rho(T_1) + \Vg_\rho(T_2)$ and $\Vg_\rho(T_1 \otimes T_2) \geq \Vg_\rho(T_1) \Vg_\rho(T_2)$.
 \item For any $T$ we have $\Vg_\omega(T) \leq \bR(T)$.
\end{enumerate}
\end{lemma}

The last item implies the asymptotic sum inequality since taking $T = \bigoplus_{i=1}^K \mm{n_i,m_i,p_i}$, the first two items show that
\[ \Vg_\omega\left(\bigoplus_{i=1}^K \mm{n_i,m_i,p_i}\right) \geq \sum_{i=1}^K (n_im_ip_i)^{\omega/3}, \]
and so the last item shows that $\sum_{i=1}^K (n_im_ip_i)^{\omega/3} \leq \bR(\bigoplus_{i=1}^K \mm{n_i,m_i,p_i})$.


\subsection{Partitioned tensors and the Coppersmith--Winograd identity} \label{sec:partitioned}
Coppersmith and Winograd~\cite{CoppersmithWinograd} exhibit the following identity, for any $q \geq 0$:
\begin{align*}
&\epsilon^3 \left[\sum_{i=1}^q \left(x_0^{[0]} y_i^{[1]} z_i^{[1]} + x_i^{[1]} y_0^{[0]} z_i^{[1]} + x_i^{[1]} y_i^{[1]} z_0^{[0]}\right) +
x_0^{[0]} y_0^{[0]} z_{q+1}^{[2]} + x_0^{[0]} y_{q+1}^{[2]} z_0^{[0]} + x_{q+1}^{[2]} y_0^{[0]} z_0^{[0]}\right] + O(\epsilon^4) = \\
&\epsilon \sum_{i=1}^q (x_0^{[0]} + \epsilon x_i^{[1]}) (y_0^{[0]} + \epsilon y_i^{[1]}) (z_0^{[0]} + \epsilon z_i^{[1]}) - \\
&\left(x_0^{[0]} + \epsilon^2\sum_{i=1}^q x_i^{[1]}\right) \left(y_0^{[0]} + \epsilon^2\sum_{i=1}^q y_i^{[1]}\right) \left(z_0^{[0]} + \epsilon^2\sum_{i=1}^q z_i^{[1]}\right) + \\
&(1 - q \epsilon)(x_0^{[0]} + \epsilon^3 x_{q+1}^{[2]}) (y_0^{[0]} + \epsilon^3 y_{q+1}^{[2]}) (z_0^{[0]} + \epsilon^3 z_{q+1}^{[2]}).
\end{align*}

This identity shows that $\bR(\tcwq) \leq q+2$, where
\[ \tcwq = \sum_{i=1}^q \left(x_0^{[0]} y_i^{[1]} z_i^{[1]} + x_i^{[1]} y_0^{[0]} z_i^{[1]} + x_i^{[1]} y_i^{[1]} z_0^{[0]}\right) +
x_0^{[0]} y_0^{[0]} z_{q+1}^{[2]} + x_0^{[0]} y_{q+1}^{[2]} z_0^{[0]} + x_{q+1}^{[2]} y_0^{[0]} z_0^{[0]}. \]
For simplicity, when $q$ is understood we use $\tcw$ for $\tcwq$. We call $\tcw$ the \emph{Coppersmith--Winograd tensor}.

The Coppersmith--Winograd tensor is an example of a partitioned tensor.

\begin{definition} \label{def:partitioned}
 Let $X,Y,Z$ be finite sets of variables, and assume that these sets are partitioned into smaller sets:
\[
X = \bigcup_{i\in I} X_i, \hspace{5mm}
Y = \bigcup_{j\in J} Y_j, \hspace{5mm}
Z = \bigcup_{k\in K} Z_k,
\]
where $I,J,K$ are three finite sets, and the unions are disjoint.
Each $X_i$ is called an \emph{$x$-group}, each $Y_j$ is called a \emph{$y$-group},
and each $Z_k$ is called a \emph{$z$-group}. Each of them is a \emph{group}.

A \emph{partitioned tensor} over $X,Y,Z$ is a tensor $T$ of the form $T = \sum_s T_s$, where each $T_s$ is a nonzero tensor over $X_{i_s},Y_{j_s},Z_{k_s}$ for some $(i_s,j_s,k_s)\in I\times J\times K$. We call $(i_s,j_s,k_s)$ the \emph{annotation} of $T_s$. The annotations of different $T_s$ must be distinct. We call $T_s$ the \emph{constituent tensors}, and $T = \sum_s T_s$ is the \emph{decomposition} of $T$. The \emph{support} of $T$ is the set $\supp(T)\subseteq I\times J\times K$ of all annotations of constituent tensors $T_s$. For convenience, we will often label the constituent tensors by elements of the support (i.e., identify $s$ and $(i_s,j_s,k_s)$) and write $T=\sum_{s\in\supp(T)}T_s$.

A partitioned tensor $T$ is \emph{tight} if $I,J,K$ are sets of integers and for some $D \in \Int$, all annotations $(i,j,k)$ in the support of $T$ satisfy $i+j+k=D$.
\end{definition}

Tightness is necessary in current techniques, based on the laser method, proving lower bounds on the value of partitioned tensors (in particular, in Theorem \ref{thm:cw-formula}).
%

As an example, we explain how to view the Coppersmith--Winograd tensor as a partitioned tensor. The $x$-variables are $X = X_0 \cup X_1 \cup X_2$, where $X_0 = \{x^{[0]}_0\}$, $X_1 = \{x^{[1]}_1,\ldots,x^{[1]}_q\}$, $X_2 = \{x^{[2]}_{q+1}\}$. The sets $Y,Z$ are defined similarly. We have
\[ T = T_{0,1,1} + T_{1,0,1} + T_{1,1,0} + T_{2,0,0} + T_{0,2,0} + T_{0,0,2}, \]
where $T_{0,1,1} \approx \mm{1,1,q}$, $T_{1,0,1} \approx \mm{q,1,1}$, $T_{1,1,0} \approx \mm{1,q,1}$ and $T_{2,0,0},T_{0,2,0},T_{0,0,2} \approx \mm{1,1,1}$. The partitioned tensor is tight since all annotations in its support sum to $2$.

Partitioned tensors can be multiplied. Suppose that $T = \sum_s T_s$ is a partitioned tensor over $X,Y,Z$, where $X = \bigcup_{i\in I} X_i$, $Y = \bigcup_{j\in J} Y_j$, $Z = \bigcup_{k\in K} Z_k$, and that $T' = \sum_{s'} T'_{s'}$ is a partitioned tensor over $X',Y',Z'$, where $X' = \bigcup_{i'\in I'} X'_{i'}$, $Y' = \bigcup_{j'\in J'} Y'_{j'}$, $Z' = \bigcup_{k'\in K'} Z'_{k'}$. Then $T \otimes T' = \sum_{s,s'} T_s \otimes T'_{s'}$ is a partitioned tensor over $X\times X',Y\times Y',Z\times Z'$, where $X\times X' = \bigcup_{(i,i') \in I\times I'} X_i \times X'_{i'}$, and $Y\times Y',Z\times Z'$ are defined similarly. If $T$ and $T'$ are both tight then so is $T\otimes T'$.

Of particular interest to us is the tensor power of a partitioned tensor. Suppose that $T = \sum_s T_s$ is a partitioned tensor over $X,Y,Z$, where $X = \bigcup_{i\in I} X_i$, $Y = \bigcup_{j\in J} Y_j$, $Z = \bigcup_{k\in K} Z_k$. For $N \geq 1$, the tensor power $T^{\otimes N}$ is a partitioned tensor over $X^N,Y^N,Z^N$. We can index the parts in $X^N$ by sequences in $I^N$ which we call \emph{$x$-indices}, and we define $y$-indices and $z$-indices analogously. Each constituent tensor of $T^{\otimes N}$ is indexed by an \emph{index triple} $(i,j,k)$ consisting of an $x$-index, a $y$-index, and a $z$-index. A set of index triples is \emph{strongly disjoint} if no two triples share an $x$-index, a $y$-index or a $z$-index.

Partitioned tensors can also be rotated: if $T$ is a partitioned tensor over $X,Y,Z$, then $\rot{T}$ is a partitioned tensor over $Y,Z,X$ (partitioned in the same way) with a rotated support. Rotation preserves tightness. A partitioned tensor $T$ with parameters $X,Y,Z,I,J,K$ is \emph{symmetric} if $I=J=K$ and $T_{(i,j,k)}$ is equivalent to $T_{(j,k,i)}$ for each $(i,j,k)\in\supp(T)$. For example, $\tcw$ is symmetric. 

The definition of value given in Section \ref{sec:value} is in terms of degeneration. However, all constructions below use a very specific form of degeneration, partitioned restriction.

\begin{definition} \label{def:partitioned-restriction}
 Let $T$ be a partitioned tensor over $X = \bigcup_{i\in I} X_i$, $Y = \bigcup_{j\in J} Y_j$, $Z = \bigcup_{k \in K} Z_k$. A \emph{partitioned restriction} of $T$ is a tensor $T'$ over $X' = \bigcup_{i \in I'} X_i$, $Y = \bigcup_{j \in J'} Y_j$, $Z = \bigcup_{k \in K'} Z_k$ (with the induced partitions), for some $I'\subseteq I$, $J'\subseteq J$, $K'\subseteq K$, obtained from $T$ by zeroing all variables in $X \setminus X',Y \setminus Y',Z \setminus Z'$.
\end{definition}

In other words, a tensor $T'$ is a \emph{partitioned restriction} of a partitioned tensor $T$ if it is obtained from $T$ by zeroing groups of variables.


\subsection{The laser method} \label{sec:laser}
The asymptotic sum inequality, or more precisely its generalization given in Section \ref{sec:value}, can be used to derive an upper bound on $\omega$ for partitioned tensors in which the set of index triples corresponding to the constituent tensors is strongly disjoint. The laser method, invented by Strassen~\cite{Strassen87}, is a general method to analyze partitioned tensors when this condition does not hold. The method has been further developed by Coppersmith and Winograd~\cite{CoppersmithWinograd}, and received its definitive form by Stothers~\cite{DavieStothers}, in the case of tight partitioned tensors. In this subsection we describe this method.

Several of the results appearing here have not appeared explicitly in prior literature, and their proofs are given in the appendix. These include Theorem~\ref{thm:cw-method}, the second half of Theorem~\ref{thm:cw-formula}, and Theorem~\ref{thm:cw-formula-exact}.

The key idea of the laser method is to obtain a lower bound on $\Vg_{\rho,N}(T)$ by considering partitioned restrictions of $T^{\otimes N}$. It is useful to abstract this idea by defining a restricted notion of value. First we define the notion of a partitioned tensor with lower bounds on the value of its constituent tensors.

\begin{definition} \label{def:cw-est}
 An \emph{estimated partitioned tensor} is a partitioned tensor $T = \sum_s T_s$ along with a function $\Val_\rho(T_s)$ for any $s \in \supp(T)$ (the \emph{estimated value}) mapping $[2,3]$ to the non-negative reals. If $T_s$ is a matrix multiplication tensor, then we insist that $\Val_\rho(T_s) = \Vol(T_s)^{\rho/3}$.
\end{definition}

If $T$ is an estimated partitioned tensor then its rotation $\rot{T}$ can be viewed as an estimated partitioned tensor by using the same values.
If $T,T'$ are estimated partitioned tensors then we can view their product $T\otimes T'$ as an estimated partitioned tensor by defining $\Val_\rho(T_s \otimes T'_{s'}) = \Val_\rho(T_s) \Val_\rho(T'_{s'})$.
We can now define the \pr value.

\begin{definition} \label{def:cw-value}
 Let $T$ be an estimated partitioned tensor. Given $\rho\in[2,3]$ and $N\geq 1$, let $\Vcw_{\rho,3N}(T)$ be the maximum of $\sum_{s \in \supp(T')} \Val_\rho(T'_s)$ over all partitioned restrictions $T'$ of $(T\otimes\rot{T}\otimes\rott{T})^{\otimes N}$ whose support is strongly disjoint.
 The \emph{\pr value} of $T$ is the function
\[ \Vcw_\rho(T) = \lim_{N\to\infty} \Vcw_{\rho,3N}(T)^{1/3N}. \]
\end{definition}

In other words, in order to compute $\Vcw_{\rho,3N}(T)$ we consider all possible ways of zeroing blocks of variables in $(T\otimes\rot{T}\otimes\rott{T})^{\otimes N}$ such that all surviving constituent tensors have distinct $x$-indices, $y$-indices and $z$-indices, and maximize over the value of $\sum_{s \in S} \Val_\rho(T'_s)$, where $S$ is the set of surviving index triples.

The idea is to choose for $\Val_\rho(T_s)$ some lower bound on $\Vg(T_s)$. For example, if $T_s$ is a matrix multiplication tensor then we can choose $\Val_\rho(T_s) = \Vol(T_s)^{\rho/3}$. A somewhat subtle application of the generalized asymptotic sum inequality then implies the following.

\begin{theorem} \label{thm:cw-method}
Let $T$ be an estimated partitioned tensor and $\rho \in [2,3]$. If $\Val_\rho(T_s) \leq \Vg_\rho(T_s)$ for all $s \in \supp(S)$ then $\Vcw_\rho(T) \leq \Vg_\rho(T)$, and in particular $\Vcw_\omega(T) \leq \bR(T)$.
\end{theorem}

Le~Gall gave a lower bound on the \pr value of tight estimated partioned tensors, which is tight in many cases. First we need to define a penalty term.

\begin{definition}[\cite{LeGall}] \label{def:dist}
 Let $T$ be a partitioned tensor over $X,Y,Z$, where $X = \bigcup_{i\in I} X_i$, $Y = \bigcup_{j\in J} Y_j$, $Z = \bigcup_{k\in K} Z_k$. The set $\dist(T)$ consists of all probability distributions over $\supp(T)$. If $T$ is symmetric then the set $\sdist(T)$ consists of all symmetric probability distributions over $\supp(T)$, that is, ones satisfying $P(i,j,k) = P(j,k,i)$ for all $(i,j,k) \in \supp(T)$.

 For a distribution $P \in \dist(T)$, its marginals to $I,J,K$ are denoted $P_1,P_2,P_3$.
 Two distributions $P,Q \in \dist(T)$ are \emph{compatible} if their marginals to $I,J,K$ are identical. The \emph{compatibility penalty} of $P \in \dist(T)$ is the quantity
\[
\Gamma_T(P) = \max_{Q} H(Q) - H(P),
\]
where the maximum is over all the distributions  $Q\in \dist(T)$ that are compatible with $P$.
\end{definition}

Note that $\Gamma_T(P) \geq 0$ always. In simple cases, two distributions $P,Q \in \dist(T)$ are compatible if and only if they are equal. This is the case for the Coppersmith--Winograd tensor, for example. For such partitioned tensors, $\Gamma_T(P) = 0$ for all $P \in \dist(T)$.
Now we can give the full theorem.

\begin{theorem}[\cite{LeGall}] \label{thm:cw-formula}
 Let $T$ be a tight estimated partitioned tensor. For any $\rho \in [2,3]$ we have
\begin{align*}
 \log \Vcw_\rho(T) &\geq \max_{P \in \dist(T)} \sum_{\ell=1}^3 \frac{H(P_\ell)}{3} + \EE_{s \sim P}[\log \Val_\rho(T_s)] - \Gamma_T(P), \\
 \log \Vcw_\rho(T) &\leq \max_{P \in \dist(T)} \sum_{\ell=1}^3 \frac{H(P_\ell)}{3} + \EE_{s \sim P}[\log \Val_\rho(T_s)].
\end{align*}
 When $T$ is symmetric, we can replace $\dist(T)$ with $\sdist(T)$, and the first summand with $H(P_1)$.
\end{theorem}

Le~Gall actually proved only the upper bound. The lower bound was proved by Cohn, Kleinberg, Szegedy and Umans~\cite{CKSU} in a special case, but their method easily extends to the general case, as shown in the appendix.

For instance, applying Theorem~\ref{thm:cw-formula} for the partitioned tensor $\tcwq[6]$ and $\Val_\rho(T_s) = \Vol(T_s)^{\rho/3}$ for all $s\in\supp(\tcwq[6])$, Coppersmith and Winograd~\cite{CoppersmithWinograd} obtained the bound $\omega < 2.3872$. Since $\Gamma_{\tcwq[6]} \equiv 0$, this bound is the optimal bound which can be obtained using the \pr value.
Theorem~\ref{thm:cw-formula} is proved by analyzing an ancillary quantity, the \pr value with respect to a distribution.

\begin{definition} \label{def:cw-value-dist}
 Let $T$ be an estimated partitioned tensor, and let $P \in \dist(T)$. For each $N$, let $\prnd{N}{P} \in \Int^{\supp(T)}$ be some vector of non-negative integers summing to $N$ obtained by canonically rounding the real vector $N\cdot P$ so that it sums to $N$. For a partitioned restriction $T'$ of $(T\otimes\rot{T}\otimes\rott{T})^{\otimes N}$, let $\supp_P(T')$ consist of all vectors in $\supp(T')$ in which for each $s \in \supp(T)$, the factors $T_s,\rot{T_s},\rott{T_s}$ (constituent tensors of $T,\rot{T},\rott{T}$, respectively) appear exactly $(\prnd{N}{P})(s)$ times each.

 Given $\rho\in[2,3]$ and $N\geq 1$, let $\Vcw_{\rho,P,3N}(T)$ be the maximum of $\sum_{s \in \supp_P(T')} \Val_\rho(T'_s)$ over all partitioned restrictions $T'$ of $(T\otimes\rot{T}\otimes\rott{T})^{\otimes N}$ whose support is strongly disjoint.
 The \emph{\pr value of $T$ with respect to $P$} is the function
\[ \Vcw_{\rho,P}(T) = \lim_{N\to\infty} \Vcw_{\rho,P,3N}(T)^{1/3N}. \]
\end{definition}

We show in the appendix that the limit always exists, and prove the following crucial property of this quantity as well.

\begin{theorem} \label{thm:cw-formula-exact}
 Let $T$ be an estimated partitioned tensor. For all $\rho\in[2,3]$ we have
\[
 \Vcw_\rho(T) = \max_{P \in \dist(T)} \Vcw_{\rho,P}(T).
\]
\end{theorem}

Theorem~\ref{thm:cw-formula} is obtained by giving lower and upper bounds for $\Vcw_{\rho,P,3N}(T)$. 

\subsection{Recursive Coppersmith--Winograd construction} \label{sec:cw-powers}

Coppersmith and Winograd~\cite{CoppersmithWinograd} obtained a better bound by considering a repartitioning $\tcw'$ of the partitioned tensor $\tcw^{\otimes 2}$, and applying the laser method to $\tcw'$. Their basic idea is to use the following partition for $X' = X^2$, and matching partitions for $Y^2$ and $Z^2$: $X' = X'_0 \cup X'_1 \cup X'_2 \cup X'_3 \cup X'_4$, where
\begin{gather*}
X'_0 = X_0 \times X_0, \qquad X'_1 = (X_0 \times X_1) \cup (X_1  \times X_0), \\
X'_2 = (X_0 \times X_2) \cup (X_1 \times X_1) \cup (X_2 \times X_0), \\
X'_3 = (X_1 \times X_2) \cup (X_2 \times X_1), \qquad X'_4 = X_2 \times X_2.
\end{gather*}
The corresponding constituent tensors come in four types:
\begin{enumerate}
 \item $\tcw'_{0,0,4} = \tcw_{0,0,2} \otimes \tcw_{0,0,2} \approx \mm{1,1,1}$.
 \item $\tcw'_{0,1,3} = \tcw_{0,1,1} \otimes \tcw_{0,0,2} + \tcw_{0,0,2} \otimes \tcw_{0,1,1} \approx \mm{1,1,2q}$.
 \item $\tcw'_{0,2,2} = \tcw_{0,1,1} \otimes \tcw_{0,1,1} + \tcw_{0,2,0} \otimes \tcw_{0,0,2} + \tcw_{0,0,2} \otimes \tcw_{0,2,0} \approx \mm{1,1,q^2+2}$.
 \item $\tcw'_{1,1,2} = \tcw_{1,1,0} \otimes \tcw_{0,0,2} + \tcw_{1,0,1} \otimes \tcw_{0,1,1} + \tcw_{0,1,1} \otimes \tcw_{1,0,1} + \tcw_{0,0,2} \otimes \tcw_{1,1,0}$.
\end{enumerate}

The last tensor $\tcw'_{1,1,2}$ is not equivalent to a matrix multiplication tensor, but it can be viewed as a tight partitioned tensor over $\bar{X}_0 \cup \bar{X}_1, \bar{Y}_0 \cup \bar{Y}_1, \bar{Z}_0 \cup \bar{Z}_1 \cup \bar{Z}_2$, where $\bar{X}_i = X_i \times X_{1-i}$, $\bar{Y}_j = Y_j \times Y_{1-j}$, and $\bar{Z}_k = Z_k \times Z_{2-k}$, with the constituent tensors corresponding to the four summands in the formula for $\tcw'_{1,1,2}$. The idea of Coppersmith and Winograd was to analyze $\tcw'$ using Theorem~\ref{thm:cw-formula}, using another application of Theorem~\ref{thm:cw-formula} to get a lower bound on the value of $\tcw'_{1,1,2}$.

More explicitly, Coppersmith and Winograd used Theorem~\ref{thm:cw-formula} to calculate $\Vcw_\rho(\tcw'_{1,1,2}) = 4^{1/3} q^\rho (2 + q^{3\rho})^{1/3}$. They then viewed $\tcw'$ itself as a partitioned tensor, with estimated values $\Val_\rho(\tcw'_{1,1,2}) = \Val_\rho(\tcw'_{1,2,1}) = \Val_\rho(\tcw'_{2,1,1}) = 4^{1/3} q^\rho (2 + q^{3\rho})^{1/3}$; all other constituent tensors $\tcw'_{i,j,k}$ are matrix multiplication tensors, and so by definition their estimated value is $\Val_\rho(\tcw'_{i,j,k}) = \Vol(\tcw'_{i,j,k})^{\rho/3}$. They then applied Theorem~\ref{thm:cw-formula} to obtain some expression for $\Vcw_\rho(\tcw')$. Theorem~\ref{thm:cw-method} shows that $\Vg_\rho(\tcw'_{1,1,2}) \geq \Val_\rho(\tcw'_{1,1,2})$, and another application of the theorem shows that $\Vcw_\omega(\tcw') \leq \bR(\tcw') = (q+2)^2$. Taking $q = 5$ and solving $\Vcw_\alpha(\tcw') = (q+2)^2$, Coppersmith and Winograd obtained the bound $\omega \leq \alpha$, where $\alpha \approx 2.3755$.

Stothers~\cite{Stothers} and Vassilevska-Williams~\cite{Williams} took this approach one step further, by considering a repartitioning $\tcw''$ of $\tcw^{\prime\otimes 2}$, along the lines of the repartitioning of $\tcw^{\otimes2}$ producing $\tcw'$ itself. The partition they use for $X''=X^{\prime2}$ is $X''_0,\ldots,X''_8$, where $X''_i = \bigcup_{i_1+i_2=i} X'_{i_1} \times X'_{i_2}$, the sum being over $0 \leq i_1,i_2 \leq 4$. Similar partitions are used for $Y''$ and $Z''$.
This time we have ten types of constituent tensors (see for example~\cite[\S5]{Williams}):
\begin{gather*}
\tcw''_{0,0,8} \approx \mm{1,1,1}, \qquad \tcw''_{0,1,7} \approx \mm{1,1,4q}, \qquad \tcw''_{0,2,6} \approx \mm{1,1,4+6q^2}, \\ \tcw''_{0,3,5} \approx \mm{1,1,12q+4q^3}, \qquad \tcw''_{0,4,4} \approx \mm{1,1,6+12q^2+q^4},  \\ \tcw''_{1,1,6}, \tcw''_{1,2,5}, \tcw''_{1,3,4}, \tcw''_{2,2,4}, \tcw''_{2,3,3}.
\end{gather*}
The first five tensors, those that contain a $0$ in their annotation, are equivalent to matrix multiplication tensors. The other five are not, and have to be analyzed like $\tcw'_{1,1,2}$ before. We illustrate the analysis using the example of $\tcw''_{1,1,6}$:
\[ \tcw''_{1,1,6} = \tcw'_{0,1,3} \otimes \tcw'_{1,0,3} + \tcw'_{1,0,3} \otimes \tcw'_{0,1,3} + \tcw'_{0,0,4} \otimes \tcw'_{1,1,2} + \tcw'_{1,1,2} \otimes \tcw'_{0,0,4}. \]
As before, we treat this as an estimated tight partitioned tensor over $\bar{X}_0 \cup \bar{X}_1, \bar{Y}_0 \cup \bar{Y}_1, \bar{Z}_2 \cup \bar{Z}_3 \cup \bar{Z}_4$, along similar lines as before. Under this partition, $\tcw''_{1,1,6}$ has four constituent tensors. The first two, $\tcw'_{0,1,3}  \otimes \tcw'_{1,0,3} \approx \tcw'_{1,0,3} \otimes \tcw'_{0,1,3} \approx \mm{2q,1,2q}$, are equivalent to matrix multiplication tensors, and we set their estimated values accordingly to $(4q^2)^{\rho/3}$. The other two, $\tcw'_{0,0,4} \otimes \tcw'_{1,1,2} \approx \tcw'_{1,1,2} \otimes \tcw'_{0,0,4}$, are more complicated, and we assign them estimated value $\Val_\rho(\tcw'_{0,0,4} \otimes \tcw'_{1,1,2}) = \Val_\rho(\tcw'_{0,0,4}) \Val_\rho(\tcw'_{1,1,2}) = 4^{1/3} q^\rho (2 + q^{3\rho})^{1/3}$, where the estimated values on the right-hand side are those of $\tcw'$. Since the value is super-multiplicative, we know that $\Vg_\rho(\tcw'_{0,0,4} \otimes \tcw'_{1,1,2}) \geq \Val_\rho(\tcw'_{0,0,4}) \Val_\rho(\tcw'_{1,1,2})$, and so this setting of the estimated value will allow us to apply Theorem~\ref{thm:cw-method} later on.

The other four complicated tensors are interpreted as estimated tight partitioned tensors in a similar fashion. For each of these, we then apply Theorem~\ref{thm:cw-formula} to compute their \pr value.\footnote{\label{fn:exact-calc} Note that we cannot use Theorem~\ref{thm:cw-formula} to calculate exactly the \pr value, for two reasons: first, the lower bound in Theorem~\ref{thm:cw-formula} is not necessarily tight; and second, the numerical optimization involved is difficult to solve optimally. Whatever quantities we get are lower bounds on the corresponding values according to Theorem~\ref{thm:cw-method}, and we use them as the estimated values of these tensors.} Applying Theorem~\ref{thm:cw-formula} and Theorem~\ref{thm:cw-method} to $\tcw''$ itself allows us to obtain an improved bound on $\omega$, namely $\omega < 2.37293$.

Vassilevska-Williams iterated this procedure once more to obtain the bound $\omega < 2.37287$, and Le~Gall~\cite{LeGall} iterated it twice more to get slightly better bounds (the improvement is only in the seventh digit after the decimal point!). We call this approach the \emph{canonical recursive Coppersmith--Winograd method}. We call the tensor obtained by iterating the construction $d$ times the \emph{canonical $\tcw^{\otimes 2^d}$}. More details on this construction appear in~\cite[\S5]{LeGall}, and the upper bounds on~$\omega$ obtained by this method are presented in Table \ref{tab:simulations-upperbounds}.

\begin{table}[h!]
\renewcommand\arraystretch{1}
\begin{center}
\caption{Upper bounds on $\omega$ obtained by analyzing $\tcw^{\otimes 2^r}$ with the canonical recursive Coppersmith--Winograd method, for several values of $r$ and $q$.}\label{tab:simulations-upperbounds}\vspace{3mm}
\begin{tabular}{c|c|c|c|c|c|c|}
&$r=0$&$r=1$&$r=2$&$r=3$&$r=4$\\
\hline
$q=1$&$3$&$2.8084$&$2.6520$&$2.6324$&$2.6312$\\
$q=2$&$2.6986$&$2.4968$&$2.4707$&$2.4690$&$2.4689$\\
$q=3$&$2.4740$&$2.4116$&$2.4030$&$2.4027$&$2.4027$\\
$q=4$&$2.4142$&$2.3838$&$2.3796$&$2.3794$&$2.3794$\\
$q=5$&$2.3935$&$2.3756$&$2.3730$&$2.3729$&$2.3729$\\
$q=6$&$2.3872$&$2.3755$&$2.3737$&$2.3737$&$2.3737$\\
$q=7$&$2.3875$&$2.3793$&$2.3780$&$2.3779$&$2.3779$\\
$q=8$&$2.3909$&$2.3848$&$2.3838$&$2.3838$&$2.3838$\\
\hline
\end{tabular}
\end{center}
\end{table}

There are two main degrees of freedom in this method. The first, mentioned already by Coppersmith and Winograd, is the method used to repartition the tensor after squaring. When squaring a tensor $T$ with partition $X_0,\ldots,X_D$, the new partition of $X' = X^2$ is into $X'_0,\ldots,X'_{2D}$, where $X_i \times X_j$ is put into $X'_{i+j}$. Coppersmith and Winograd suggest trying out other merging schemes. While in order to apply Theorem~\ref{thm:cw-formula} we need the resulting repartitioning to be tight, on both $T^{\otimes 2}$ and its constituent tensors, one can conceive other repartitionings which could be analyzed differently (but still using the \pr value). The second degree of freedom, which is also mentioned by Coppersmith and Winograd, suggests a different choice of which tensors to multiply each time. The canonical method starts with $\tcw$, computes $\tcw' = \tcw \otimes \tcw$, then $\tcw'' = \tcw' \otimes \tcw'$, and so on. However, instead of choosing $\tcw'' = \tcw' \otimes \tcw'$, we could have chosen $\tcw'' = \tcw' \otimes \tcw$, a choice which would correspond to an analysis of $\tcw^{\otimes 3}$. Calculation reveals that analyzing $\tcw^{\otimes 3}$ does not result in improved bounds, but it is possible that other multiplication schemes would be advantageous.

The method we describe in Section~\ref{sec:merging} will subsume all such multiplication schemes on~$\tcw$. This method actually works on a larger class of multiplication schemes which we now describe formally. Since the description of this class of schemes is not specific to $\tcw$, the presentation below is given for arbitrary estimated partitioned tensors $T$ and $T'$ (in our applications, both $T$ and $T'$ will be powers of $\tcw$).

We start with the concept of repartitioning of a partitioned tensor.

\begin{definition} \label{def:repartitioning}
Let $T$ be a symmetric estimated partitioned tensor over $X,Y,Z$ partitioned as
\[
X = \bigcup_{i\in I} X_i, \hspace{3mm} Y = \bigcup_{j\in J} Y_j, \hspace{3mm}Z = \bigcup_{k\in K} Z_k,
\]
 and let $T'$ be a symmetric estimated partitioned tensor over $X',Y',Z'$ partitioned as
 \[
 X = \bigcup_{i' \in I'} X_{i'}, \hspace{3mm}Y' = \bigcup_{j'\in J'} Y_{j'}, \hspace{3mm}Z' = \bigcup_{k'\in K'} Z_{k'}.
 \]
We also assume that the value estimates of $T$ and $T'$ are symmetric, that is, for each constituent tensor $T_s$ of $T$, $\Val_\rho(T_s) = \Val_\rho(\rot{T_s})$, and similarly for $T'$.

A \emph{repartitioning} $\tilde{T}$ of $T \otimes T'$ is a partitioned tensor over $\tilde{X}=X\times X',\tilde{Y}=Y\times Y',\tilde{Z}=Z\times Z'$ satisfying the following properties:
\begin{itemize}
\item
the partition of $\tilde{X}$ is a coarsening of the partition $\bigcup_{(i,i') \in I\times I'} X_i \times X'_{i'}$,
\item
the partition of $\tilde{Y}$ is a coarsening of the partition $\bigcup_{(j,j') \in J\times J'} Y_j \times Y'_{j'}$,
\item
the partition of $\tilde{Z}$ is a coarsening of the partition $\bigcup_{(k,k') \in K\times K'} Z_k \times Z'_{k'}$.
\end{itemize}
\end{definition}
Note that this definition, applied to the case where both $T$ and $T'$ are powers of $\tcw$, captures the two degrees of freedom discussed above.

Let $\tilde{T}$ be a repartitioning of $T \otimes T'$. Let us use the notations of Definition \ref{def:repartitioning}, and write the corresponding coarser partitions of $\tilde{X}$, $\tilde{Y}$
and~$\tilde{Z}$ as
\[
\tilde{X} = \bigcup_{\tilde{\imath} \in \tilde{I}} \tilde{X}_{\tilde{\imath}}, \hspace{3mm}
\tilde{Y} = \bigcup_{\tilde{\jmath} \in \tilde{J}} \tilde{Y}_{\tilde{\jmath}},\hspace{3mm}
\tilde{Z} = \bigcup_{\tilde{k} \in \tilde{K}} \tilde{Z}_{\tilde{k}}.
\]
For each $\tilde{\imath} \in \tilde{I}$, we denote by $x(\tilde{\imath})$ the subset of  $I\times I'$ such that $\tilde{X}_{\tilde{\imath}} = \bigcup_{(i,i') \in x(\tilde{\imath})} X_i \times X'_{i'}$. We use similar notations for the partitions of $\tilde{Y}$ and $\tilde {Z}$. The idea to derive an upper bound on $\omega$ is, again, to consider $\tilde{T}$ as an estimated partitioned tensor. To do this, we need to define $\Val_\rho(\tilde{T}_{\tilde{\imath},\tilde{\jmath},\tilde{k}})$ for each $(\tilde{\imath},\tilde{\jmath},\tilde{k}) \in \supp(\tilde{T})$. Remember that, since $T$ and $T'$ are estimated partitioned tensors, $\Val_\rho(T_s)$ and $\Val_\rho(T'_{s'})$ are given for the constituent tensors of $T$ and $T'$. We use these quantities to define $\Val_\rho(\tilde{T}_{\tilde{\imath},\tilde{\jmath},\tilde{k}})$, as follows. If $\tilde{T}_{\tilde{\imath},\tilde{\jmath},\tilde{k}}$ is equivalent to a matrix product, we set
\begin{equation}
\Val_\rho(\tilde{T}_{\tilde{\imath},\tilde{\jmath},\tilde{k}})=\Vol(\tilde{T}_{\tilde{\imath},\tilde{\jmath},\tilde{k}})^{\rho/3}.
\end{equation}
Otherwise, we consider $\tilde{T}_{\tilde{\imath},\tilde{\jmath},\tilde{k}}$ itself as an estimated partitioned tensor,
with support
\[
\left\{ ((i,i'),(j,j'),(k,k')) \in x(\tilde{\imath})\times y(\tilde{\jmath})\times z(\tilde{k})
\:\big|\: (i,j,k) \in \supp(T) \textrm{ and }(i',j',k') \in \supp(T')\right\}
\]
and constituent tensors $T_{i,j,k} \otimes T'_{i',j',k'}$ and estimated value $\Val_\rho(T_{i,j,k} \otimes T'_{i',j',k'}) = \Val_\rho(T_{i,j,k}) \times \Val_\rho(T'_{i',j',k'})$ for each constituent tensor, and set
\begin{equation}\label{eq:Valo}
\Val_\rho(\tilde{T}_{\tilde{\imath},\tilde{\jmath},\tilde{k}})= \Vcw_\rho(\tilde{T}_{\tilde{\imath},\tilde{\jmath},\tilde{k}}).
\end{equation}
The resulting tensor $\tilde{T}$ is a symmetric estimated partitioned tensor with symmetric value estimates.

The process we have just described is exactly what is done in the canonical Coppersmith--Winograd method (where both~$T$ and $T'$ are powers of $\tcw$ and the repartitioning sums the two indices of the variables). As was done there, we can apply Theorem~\ref{thm:cw-formula} to compute this partition-restricted value, or a lower bound on it.


We can naturally iterate recursively the above construction, which leads to the following definition.
\begin{definition} \label{def:recursive-repartitioning}
Let $T$ be an estimated partitioned tensor. An estimated partitioned tensor $T'$ is a \emph{recursive repartitioning} of $T$ if there exists a sequence $T_1,\ldots,T_\ell$ such that (i) $T_1 = T$, (ii) $T_\ell = T'$, (iii) for each $i > 1$, there exist $j,k < i$ such that $T_i$ is a repartitioning of $T_j \otimes T_k$.
\end{definition}


In particular, the canonical $\tcw^{\otimes 2^N}$ is a recursive repartitioning of the canonical $\tcw^{\otimes 2^M}$ for all $M \leq N$.
One of the main technical contribution of this paper is defining a notion of value $\Vm_\rho$, in the next section, that satisfies $\Vm_\rho(\tcw^{\otimes D}) \geq \Vcw_\rho(\tcw^{\otimes dD})^{1/d}$ whenever $\tcw^{\otimes dD}$ is a recursive repartitioning of $\tcw^{\otimes D}$.

\section{Merging} \label{sec:merging}
Theorem~\ref{thm:cw-method} allows us to obtain upper bounds on $\omega$ by analyzing powers of the Coppersmith--Winograd tensor. Experimentally, fixing $q$ we find out that the bound on $\omega$ obtained by considering the canonical $\tcw^{\otimes 2^\ell}$ improves as $\ell$ increases, but the root cause of this phenomenon has never been completely explained. Indeed, as mentioned in the introduction, at first glance it seems that considering powers of $\tcw$ should not help at all, since our analysis proceeds by analyzing powers $\tcw^{\otimes N}$ for large $N$; how do we gain anything by analyzing instead large powers of $\tcw^{\otimes 2}$? The improvement results from the fact that when defining $\tcw^{\otimes 2}_s$ for annotations $s$ containing a zero, we merge together several matrix multiplication tensors into one large matrix multiplication tensors.
Inspired by this observation, we define a notion of value which allows merging of matrix multiplication tensors, and show that the method it corresponds to subsumes the analysis of \emph{all} powers of $\tcw$.

The definition is somewhat complicated to allow analysis of powers of $\tcw$, and becomes much simpler when analyzing $\tcw$ itself, or any other tensor whose constituent tensors are all matrix multiplication tensors. For this reason, we start with a simplified version of the definition which only applies to tensors of the latter form, and only then present the general definition.

\begin{definition} \label{def:merging-special}
Let $T$ be a symmetric partitioned tensor, each of whose constituent tensors is a matrix multiplication tensor. For $N \geq 1$, we say that $T'$ is a \emph{consistent restriction} of $T^{\otimes N}$ if for some partitioned restriction $R$ of $T^{\otimes N}$, the following hold:
\begin{enumerate}
\item Each constituent tensor of $T'$ is a sum of constituent tensors of $R$, each constituent tensor of $R$ appearing exactly once as a summand in some constituent tensor of $T'$.
\item Each constituent tensor in $T'$ is equivalent to a matrix multiplication tensor.
\item Distinct constituent tensors of $T'$ have disjoint sets of $x$-variables, $y$-variables and $z$-variables.
\end{enumerate}
Given $\rho \in [2,3]$ and $N \geq 1$, we define $\Vm_{\rho,N}(T)$ to be the maximum of $\sum_{s \in \supp(T')} \Vol(T'_s)^{\rho/3}$ over all consistent restrictions $T'$ of $T^{\otimes N}$.
 The \emph{merging value} of $T$ is the function
\[ \Vm_\rho(T) = \lim_{N\to\infty} \Vm_{\rho,N}(T)^{1/N}. \]
(We show below that the limit exists.)
\end{definition}

In the general case, the definition of \emph{consistent restriction} is somewhat more complicated, since we want to put some restriction on the sets of constituent tensors in $R$ that we allow to merge: we want the non-matrix multiplication tensors to be opaque.
This prompts the following definition.

\begin{definition} \label{def:consistent-sum}
 Let $T$ be a symmetric estimated partitioned tensor, and fix $N \geq 1$. Let $\suppz(T)$ consist of all $s \in \supp(T)$ such that $T_s$ is a matrix multiplication tensor, and let $\suppn(T) = \supp(T) \setminus \suppz(T)$. A \emph{pattern} is a mapping $\pi\colon [N] \to \suppn(T) \cup \{0\}$. A constituent tensor $T_s = \bigotimes_{i=1}^n T_{s_i}$ \emph{conforms} to the pattern $\pi$ if $T_{s_i} = T_{\pi(i)}$ if $\pi(i) \in \suppn(T)$, and $T_{s_i} \in \suppz(T)$ if $\pi(i) = 0$. If $T_s$ conforms to some pattern $\pi$ then we define $\tz{T_s} = \bigotimes_{i\colon \pi(i) = 0} T_{s_i}$ and $\tn{T_s} = \bigotimes_{i\colon \pi(i) \neq 0} T_{s_i}$.

A sum $S = \sum_s T^{\otimes N}_s$ of constituent tensors of $T^{\otimes N}$ is \emph{consistent} if (i) for some pattern $\pi$, all tensors conform to $\pi$, and (ii) $\sum_s \tz{T_s}$ is equivalent to a matrix multiplication tensor, denoted $\tz{S}$. If $S$ is consistent with respect to $\pi$ then we define, for all $\rho \in [2,3]$,
\[ \Val_\rho(S) = \Vol(\tz{S})^{\rho/3} \Val_\rho(\tn{S}), \]
where $\Val_\rho(\tn{S}) = \prod_{i\colon \pi(i) \neq 0} \Val_\rho(T_{\pi(i)})$.
\end{definition}

We can now present the general definition.

\begin{definition} \label{def:merging}
Let $T$ be an estimated partitioned tensor. For $N \geq 1$, we say that $T'$ is a \emph{consistent restriction} of $T^{\otimes N}$ if for some partitioned restriction $R$ of $T^{\otimes N}$, the following hold:
\begin{enumerate}
\item Each constituent tensor of $T'$ is a consistent sum of constituent tensors of $R$, each constituent tensor of $R$ appearing exactly once as a summand in some constituent tensor of $T'$.
\item Distinct constituent tensors of $T'$ have disjoint sets of $x$-variables, $y$-variables and $z$-variables.
\end{enumerate}
Given $\rho \in [2,3]$ and $N \geq 1$, we define $\Vm_{\rho,N}(T)$ to be the maximum of $\sum_{s \in \supp(T')} \Val_\rho(T'_s)$ over all consistent restrictions $T'$ of $T^{\otimes N}$.
 The \emph{merging value} of $T$ is the function
\[ \Vm_\rho(T) = \lim_{N\to\infty} \Vm_{\rho,N}(T)^{1/N}. \]
(We show below that the limit exists.)
\end{definition}

Note that we have only defined the merging value for symmetric tensors, since the tensors $\tcw^{\otimes N}$ are all symmetric. Previously, values of asymmetric versions came up only because we calculated the \pr value recursively. In contrast, the merging value is calculated by a single application of the definition. 

We now prove several simple properties of the merging value, starting with the proof that $\Vm_\rho(T)$ is well-defined.

\begin{lemma} \label{lem:merging-well-defined}
 Let $T$ be a symmetric estimated partitioned tensor, let $\rho \in [2,3]$, and let $N \geq 1$. The limit $\lim_{N\to\infty} \Vm_{\rho,N}(T)^{1/N}$ exists.
\end{lemma}
\begin{proof}
 By Fekete's lemma, it is enough to show that $\Vm_{\rho,N_1+N_2}(T) \geq \Vm_{\rho,N_1}(T) \Vm_{\rho,N_2}(T)$. Indeed, given a consistent restriction $T'_1$ of $T^{\otimes N_1}$ and a consistent restriction $T'_2$ of $T^{\otimes N_2}$, it is not hard to construct a consistent restriction $T'_1\otimes T'_2$ of $T^{\otimes (N_1+N_2)}$ satisfying
\[ \sum_{(s_1,s_2) \in \supp(T'_1 \otimes T'_2)} \Val_\rho(T'_{1;s_1} \otimes T'_{2;s_2}) = \sum_{s_1 \in \supp(T'_1)} \Val_\rho(T'_{1;s_1}) \times \sum_{s_2 \in \supp(T'_2)} \Val_\rho(T'_{2;s_2}), \]
showing that $\Vm_{\rho,N_1+N_2}(T) \geq \Vm_{\rho,N_1}(T) \Vm_{\rho,N_2}(T)$.
\end{proof}

%

Next, the merging value is super-multiplicative and super-additive. Our proof follows a similar proof for the value due to Stothers~\cite{Stothers,DavieStothers}.

\begin{lemma} \label{lem:merging-super-mult}
 For any two symmetric estimated partitioned tensors $T_1,T_2$ we have $\Vm_\rho(T_1\otimes T_2) \geq \Vm_\rho(T_1) \Vm_\rho(T_2)$ and $\Vm_\rho(T_1\oplus T_2) \geq \Vm_\rho(T_1) + \Vm_\rho(T_2)$.
\end{lemma}
\begin{proof}
 The first part follows by the observation that $\Vm_{\rho,N}(T_1\otimes T_2) \geq \Vm_{\rho,N}(T_1) \Vm_{\rho,N}(T_2)$.

 For the second part, notice that the tensor $(T_1 \oplus T_2)^{\otimes N}$ decomposes as
\[
 (T_1 \oplus T_2)^{\otimes N} = \bigoplus_{N_1+N_2 = N} \binom{N}{N_1} \odot T_1^{\otimes N_1} T_2^{\otimes N_2},
\]
 where $M \odot T$ denotes the direct sum of $M$ tensors equivalent to $T$. In particular, 
\[
 \Vm_{\rho,N}(T_1 \oplus T_2) \geq \sum_{N_1+N_2 = N} \binom{N}{N_1} \Vm_{\rho,N_1}(T_1) \Vm_{\rho,N_2}(T_2).
\]
 Let $\alpha_1 = \frac{\Vm_\rho(T_1)}{\Vm_\rho(T_1)+\Vm_\rho(T_2)}$ and $\alpha_2 = \frac{\Vm_\rho(T_2)}{\Vm_\rho(T_1)+\Vm_\rho(T_2)}$. Considering $N_1 \approx \alpha_1 N$ and $N_2 \approx \alpha_2 N$, we get
\[
 \Vm_{\rho,N}(T_1 \oplus T_2) \gtrapprox 2^{H(\alpha_1,\alpha_2)N} \Vm_\rho(T_1)^{\alpha_1 N} \Vm_\rho(T_2)^{\alpha_2 N} \approx (\Vm_\rho(T_1) + \Vm_\rho(T_2))^N,
\]
 where the approximations are true up to polynomial factors and in the limit $N\to\infty$. This shows that $\Vm_\rho(T_1 \oplus T_2) \geq \Vm_\rho(T_1) + \Vm_\rho(T_2)$.
\end{proof}

Finally, we show that the merging value is indeed a lower bound on the value.

\begin{theorem} \label{thm:merging-method}
Let $T$ be an estimated partitioned tensor and $\rho \in [2,3]$. If $\Val_\rho(T_s) \leq \Vg_\rho(T_s)$ for all $s \in \supp(S)$ then $\Vm_\rho(T) \leq \Vg_\rho(T)$, and in particular $\Vm_\omega(T) \leq \bR(T)$.
\end{theorem}
\begin{proof}[\Proofsketch]
 When all constituent tensors are matrix multiplication tensors (or even arbitrary symmetric tensors), the proof is a straightforward application of the generalized asymptotic sum inequality, Lemma~\ref{lem:value-properties}. The proof in the general case follows the ideas of the proof of Theorem~\ref{thm:cw-method}, and involves generalizing Definition~\ref{def:cw-value-dist}. \ifsketch \highlight{Details omitted.} \fi
\end{proof}

\subsection{Merging subsumes recursive repartitioning} \label{sec:subsumption}
As explained in the beginning of this section, the gainings resulting from considering powers of the Coppersmith--Winograd tensor originate in the merging of matrix multiplication tensors during the repartitioning steps. The merging value is a more general way of doing such mergings, as we show in this subsection.

The key result is the following theorem, which shows that the square of the merging value of~$\tcw$ is an upper bound on the partitioned-value of \emph{any} repartitioning of $\tcw\otimes \tcw$ (as defined in Section~\ref{sec:cw-powers}).
\begin{theorem} \label{thm:subsumption-special}
Fix a value for $q$ and write $\tcw=\tcw(q)$. If $\tcw^{\otimes 2}$ is a repartitioning of $\tcw\otimes \tcw$ then for all $\rho \in [2,3]$ we have $\Vm_\rho(\tcw) \geq \Vcw_\rho(\tcw^{\otimes 2})^{1/2}$.
\end{theorem}
Before proving the theorem, let us note that this theorem shows that an upper bound on $\Vm(\tcw)$ implies a limit on the bound on $\omega$ achievable by analyzing any repartitioning of $\tcw\otimes \tcw$. Indeed, a bound on $\omega$ obtained using such a repartioning has the form $\omega \leq \rho$ where $\Vcw_\rho(\tcw\otimes \tcw) = (q+2)^{2}$. If $\Vm_\rho(\tcw) \leq B_{\rho}$ then $B_\rho \geq \Vm_\rho(\tcw) \geq \Vcw_\rho(\tcw\otimes \tcw)^{1/2} = (q+2)$, and so $\rho \geq \alpha$ where $\alpha$ is the solution to $B_\alpha = (q+2)$.

The idea behind the proof is simple: we give a lower bound on $\Vm_{\rho,N}(\tcw^{\otimes 2})$ by recursively applying Theorem~\ref{thm:cw-formula-exact}. The first application is to $\tcw^{\otimes 2}$ itself: Theorem~\ref{thm:cw-formula-exact} gives us a partitioned restriction of $(\tcw^{\otimes 2})^{\otimes 3N}$. However, not all constituent tensors of $\tcw^{\otimes 2}$ are mergings of constituent tensors of $\tcw$, and in order to analyze those we need another application of Theorem~\ref{thm:cw-formula-exact}.

\begin{proof}[Proof of Theorem \ref{thm:subsumption-special}]
In order to avoid confusion, we will use $\Val$ for the estimated value of constituent tensors of $\tcw$ and its powers, and $\Valo$ for the estimated value of constituent tensors of $\tcw^{\otimes 2}$ and its powers. The estimated values of constituent tensors of $\tcw_{1,1,2} \otimes \tcw_{1,2,1} \otimes \tcw_{2,1,1}$ and its powers are also given by $\Val$, by definition.

Observe that, for every $n\ge 0$,  from Definition \ref{def:cw-value} we know that there exists a
partitioned restriction $T_{3n}$ of $(\tcw_{1,1,2} \otimes \tcw_{1,2,1} \otimes \tcw_{2,1,1})^{\otimes n}$ having strongly disjoint support such that the equality $\sum_{s \in \supp(T_{3n})} \Val_{\rho}(T_{3n,s}) = \Vcw_{\rho,3n}(\tcw_{1,1,2})$ holds. Note that
\[ \sum_{s \in \supp(T_{3n})} \Val_{\rho}(T_{3n,s}) = \Vcw_{\rho,3n}(\tcw_{1,1,2}) \approx \Vcw_\rho(\tcw_{1,1,2})^{3n} = \Valo_{\rho}((\tcw_{1,1,2} \otimes \tcw_{1,2,1} \otimes \tcw_{2,1,1})^{\otimes n}), \]
where the second (approximate) equality comes from the fact that the quantity $\Vcw_{\rho,3n}(\tcw_{1,1,2})$ converges to $\Vcw_\rho(\tcw_{1,1,2})^{3n}$ when $n$ goes to infinity, and the third equality comes from the definition of $\Valo_{\rho}$ (see Equation (\ref{eq:Valo})), using the fact that $\Valo_\rho$ is symmetric.

Let $P\in\dist(\tcw^{\otimes 2})$ be the distribution satisfying $\Vcw_\rho(\tcw^{\otimes 2}) = \Vcw_{\rho,P}(\tcw^{\otimes 2})$ given by Theorem~\ref{thm:cw-formula-exact}.
For every $N$, there exists a partitioned restriction
$U_{3N}$ of $(\tcw^{\otimes 2})^{\otimes 3N}$ having strongly disjoint support such that
$\sum_{s \in \supp_P(U_{3N})} \Valo_{\rho}(U_{3N,s})
=\Vcw_{\rho,P,3N}(\tcw^{\otimes 2})$.
Each constituent tensor of $U_{3N}$ whose annotation belongs to $\supp_P(U_{3N})$ is a tensor power of $3N$ constituent tensors of $\tcw^{\otimes 2}$, which after rearrangement will have the form
\[
T^0 \otimes (\tcw_{1,1,2} \otimes \tcw_{1,2,1} \otimes \tcw_{2,1,1})^{\otimes n}
\]
for some integer $n$, where $T^0$ is a product of matrix multiplication constituent tensors of $\tcw^{\otimes 2}$
(both~$n$ and $T^0$ are the same, up to equivalence, for all constituent tensors of $U_{3N}$). The reason $\tcw_{1,1,2},\tcw_{1,2,1},\tcw_{2,1,1}$ all have the same power is that $\supp_P(U_{3N})$ counts only constituent tensors in which $\tcw_{1,1,2},\rot{\tcw_{1,1,2}},\rott{\tcw_{1,1,2}}$ all appear the same number of times.

The idea is to replace the factor $(\tcw_{1,1,2} \otimes \tcw_{1,2,1} \otimes \tcw_{2,1,1})^{\otimes n}$ with a copy of $T_{3n}$, which can be done by zeroing groups of variables in $\tcw^{\otimes 6N}$. We repeat this process with each constituent tensor in~$U_{3N}$. A crucial observation is that the zeroing operations we do on one constituent tensor have no impact on the other constituent tensors, since the constituent tensors in $U_{3N}$ have strongly disjoint supports. This construction thus gives a partitioned restriction $W_{6N}$ of $\tcw^{\otimes 6N}$ having strongly disjoint support.
We have
\begin{align*}
\sum_{s \in \supp(W_{6N})} \Val_{\rho}(W_{6N,s}) &=
|\supp(U_{3N})|
\Vol(T_0)^{\rho/3}
\times \sum_{s \in \supp(T_{3n})} \Val_{\rho}(T_{3n,s})
\\ &=
|\supp(U_{3N})|
\Vol(T_0)^{\rho/3}
\Valo_{\rho}((\tcw_{1,1,2} \otimes \tcw_{1,2,1} \otimes \tcw_{2,1,1})^{\otimes n})
\\ &= \Vcw_{\rho,P,3N}(\tcw^{\otimes 2}).
\end{align*}
We conclude that $\Vm_{\rho,6N}(\tcw)\ge \Vcw_{\rho,P,3N}(\tcw^{\otimes 2})$, and thus $\Vm_{\rho}(\tcw)\ge \Vcw_{\rho,P}(\tcw^{\otimes 2})^{1/2} = \Vcw_\rho(\tcw^{\otimes 2})^{1/2}$.
\end{proof}

Theorem \ref{thm:subsumption-special} can be easily generalized to estimated partitioned tensors other than $\tcw$,
and also to recursive partitioning (the main difference when analyzing recursive partitioning is the higher number of levels of recursion). In particular, we obtain the following result.

\begin{theorem} \label{thm:subsumption}
Fix a value for $q$ and write $\tcw=\tcw(q)$. If $\tcw^{\otimes dD}$ is a recursive repartitioning of $\tcw^{\otimes D}$ then for all $\rho \in [2,3]$ we have $\Vm_\rho(\tcw^{\otimes D}) \geq \Vcw_\rho(\tcw^{\otimes dD})^{1/d}$.
\end{theorem}

Similarly to the observation done above, but more generally, Theorem \ref{thm:subsumption} shows that an upper bound on $\Vm(\tcw^{\otimes D})$ implies a limit on the bound on $\omega$ achievable by analyzing recursive repartitionings of $\tcw^{\otimes D}$. The reason is that a bound on $\omega$ obtained using recursive repartitioning has the form $\omega \leq \rho$ where $\Vcw_\rho(\tcw^{\otimes dD}) = (q+2)^{dD}$. If $\Vm_\rho(\tcw^{\otimes D}) \leq B_{\rho}$ then $B_\rho \geq \Vm_\rho(\tcw^{\otimes D}) \geq \Vcw_\rho(\tcw^{\otimes dD})^{1/d} = (q+2)^D$, and so $\rho \geq \alpha$ where $\alpha$ is the solution to $B_\alpha = (q+2)^D$.
Moreover, Theorem~\ref{thm:merging-method} shows that every bound which can be obtained by analyzing a recursive repartitioning of $\tcw$ can also be obtained using the inequality $\Vm_\omega(\tcw) \leq q+2$. In this sense, the laser method with merging subsumes the recursive laser method.

\section{Upper bound on the value with merging} \label{sec:upper-bound}

Theorem~\ref{thm:subsumption} prompts us to obtain upper bounds on the merging value of recursive repartitionings of $\tcw$. Our approach will only use some properties of recursive repartitionings of $\tcw$, described in the following definition.

\begin{definition} \label{def:cw-like}
A partitioned tensor $T$ over $X,Y,Z$ is \emph{Coppersmith--Winograd-like} if
\begin{enumerate}
 \item $X$ is partitioned into parts $X_0,\ldots,X_D$, where $|X_0| = |X_D| = 1$ and $|X_i| > 1$ for $0 < i < D$. Similarly for $Y,Z$.
 \item Tensors annotated $(\alpha,\beta,\gamma)$ for $\alpha,\beta,\gamma \neq 0$ are not equivalent to matrix multiplication tensors. 
 \item If $(\alpha,\beta,0) \in \supp(T)$ then $|X_\alpha| = |Y_\beta| \triangleq m$ and $T_{\alpha,\beta,0} = \sum_{i=1}^m x_i y_i z$ where $X_\alpha = \{x_1,\ldots,x_m\}$, $Y_\beta = \{y_1,\ldots,y_m\}$ and $Z_0 = \{z\}$\footnote{The enumerations of $X_\alpha,Y_\beta$ can potentially depend on the annotation $(\alpha,\beta,0)$, though this does not happen in our applications.}. Similarly for tensors annotated $(\alpha,0,\beta)$ and $(0,\alpha,\beta)$.
 \item The only annotations in $\supp(T)$ involving only $0$ and $D$ are $(D,0,0),(0,D,0),(0,0,D)$.
\end{enumerate}
\end{definition}

When $q > 1$, it is not hard to check that all recursive repartitionings of $\tcw$ are Coppersmith--Winograd-like.

The rest of this section is organized as follows. In Section~\ref{sec:consistent-sum} we prove a combinatorial lemma describing when constituent tensors of $T^{\otimes N}$ can combine to a matrix multiplication tensors, for any Coppersmith--Winograd-like tensor $T$. Using this lemma, we prove an upper bound on~$\Vm_\rho(\tcw)$ in Section~\ref{sec:upper-bound-special}, and show how to extend it to general Coppersmith--Winograd-like tensors $T$ in Section~\ref{sec:upper-bound-general}. We apply the upper bound to $\tcwq^{\otimes 2^r}$ for various values of $q$ and $r$ in Section~\ref{sec:simulations}.

\subsection{Structure of consistent sums} \label{sec:consistent-sum}
The following combinatorial lemma (Lemma \ref{lem:coherent-sums}) identifies the structure of consistent sums in the case of Coppersmith--Winograd-like tensors.

\begin{definition} \label{def:coherent}
 For a partitioned tensor $T$, a \emph{zero-sequence} of length $N$ is a constituent tensor in $T^{\otimes N}$ whose index triple is in $\suppz(T)^N$, that is, its index triple $(A,B,C)$ satisfies the property that for $i \in [N]$, one of $A_i,B_i,C_i$ is zero.

 A sum of distinct zero-sequences is \emph{consistent} if it is equivalent to a matrix multiplication tensor. It is \emph{coherent} if there is a partition $[N] = X \cup Y \cup Z$ such that each index triple $(A,B,C)$ of a tensor in the sum satisfies $A_i = 0$ for all $i \in X$, $B_j = 0$ for all $j \in Y$, and $C_k = 0$ for all $k \in Z$.
\end{definition}

\begin{lemma} \label{lem:coherent-sums}
 Let $T$ be a Coppersmith--Winograd-like symmetric partitioned tensor.
 If the sum of distinct zero-sequences of length $N$ of constituent tensors of $T$ is consistent, then it is coherent.
\end{lemma}
\begin{proof}
 We can identify a zero-sequence of length $N$ with a vector in $\suppz(T)^N$, and so the set of distinct zero-sequences with a subset $S \subseteq \suppz(T)^N$. We will also think of $\suppz(T)^N$ as a subset of $(\{0,\ldots,D\}^N)^3$, the latter being the set of all index triples. We can write the sum itself as $\sum_{s \in S} T_s$, where $T_s$ is the unique constituent tensor of $T^{\otimes N}$ which corresponds to the index triple $s$.
 Our goal is to show that each $i \in [N]$ is either $x$-constant (all $(A,B,C) \in S$ satisfy $A_i = 0$), $y$-constant (same for $B_i = 0$), or $z$-constant (same for $C_i = 0$).

 Suppose that this sum is equivalent to the matrix multiplication tensor $\mm{n,m,p}$.
 Recall that
\[ \mm{n,m,p} = \sum_{i=1}^n \sum_{j=1}^m \sum_{k=1}^p x_{ij} y_{jk} z_{ki}. \]
 Since $\sum_{s \in S} T_s \approx \mm{n,m,p}$, there is a bijection between the $x$-variables appearing in $\sum_{s \in S} T_s$ and the $nm$ variables $x_{ij}$, and similar bijections for the $y$-variables and $z$-variables, such that applying these bijections to $\sum_{s \in S} T_s$, we obtain exactly $\mm{n,m,p}$. Fix such bijections, which we call \emph{denotations}.

 Let $\ell(w) = |X_w| = |Y_w| = |Z_w|$, and denote the variables in $X_w$ by $\hat{x}^{[w]}_1,\ldots,\hat{x}^{[w]}_{\ell(w)}$. Each $x$-variable appearing in $T^{\otimes N}$ belongs to some group $A$, and the $x$-variables of group $A$ can be indexed by vectors $v$ of length $N$ such that $v_s \in [\ell(A_s)]$ for all $s \in [N]$. Call two $x$-variables belonging to the same group \emph{$s$-siblings} if they differ only in $v_s$.

 Consider some variable $\hat{x}^A_u$ and some $t \in [N]$ such that $0 < A_t < D$, and suppose that $\hat{x}^A_u$ appears in some $(A,B,C) \in S$ satisfying $C_t = 0$, say as part of the product $\hat{x}^A_u \hat{y}^B_v \hat{z}^C_w$. Denote the $t$-siblings of $\hat{x}^A_u$ and $\hat{y}^B_v$ by $\hat{X}_r,\hat{Y}_r$ in such a way that $T_{(A,B,C)}$ includes the sum $\sum_r \hat{X}_r \hat{Y}_r \hat{z}^C_w$, where $r$ ranges over $[\ell(A_t)] = [\ell(B_t)]$; this is possible since the $t$th factor in $T_{(A,B,C)}$ has the form $\sum_r \hat{x}^{[A_t]}_r \hat{y}^{[B_t]}_{\sigma(r)} \hat{z}^{[0]}_1$ for some permutation $\sigma$. If the denotation of $\hat{z}^C_w$ is $z_{ki}$ then the denotations of $\hat{X}_r$ are all in row $i$, and the denotations of $\hat{Y}_r$ are all in column $k$. In particular, the denotations of the $t$-siblings of $\hat{x}^A_u$ are in the same row. In contrast, had we assumed that $B_t = 0$ instead of $C_t = 0$, we would have concluded that the denotations of the $t$-siblings of $\hat{x}^A_u$ are in the same column rather than row.

 Call a $z$-variable $z_{ki}$ \emph{$t$-good} if it appears in the tensor corresponding to some index triple $(A,B,C) \in S$ satisfying $C_t = 0$ and $0 < A_t,B_t < D$, say as part of the product $x_{ij} y_{jk} z_{ki}$. The foregoing shows that the $t$-siblings of $x_{ij}$ are all in the same row. Now pick an arbitrary $K \in [p]$. The product $x_{ij} y_{jK} z_{Ki}$ must appear in the tensor corresponding to some $(A,B',C') \in S$ (note that the first index is the same). Since $A_t \neq 0$, either $B_t = 0$ or $C_t = 0$. In the former case, the foregoing would imply that the $t$-siblings of $x_{ij}$ are all in the same column, leading to a contradiction (since $\ell(A_t) > 1$), and we conclude that $C_t = 0$ and so $z_{Ki}$ is also $t$-good. We can similarly change $i$ to any $I \in [n]$, showing that all $z$-variables are $t$-good.

 Consider now an arbitrary $t \in [N]$. If $S$ contains some index triple $(A,B,C)$ such that $(A_t,B_t,C_t) \notin \{(D,0,0),(0,D,0),(0,0,D)\}$, then the above argument shows that it is constant (i.e., either $x$-constant, $y$-constant, or $z$-constant). It remains to consider the case in which all index triples $(A,B,C)$ appearing in $S$ satisfy $(A_t,B_t,C_t) \in \{(D,0,0),(0,D,0),(0,0,D)\}$. If at most two of these actually appear then again $t$ is constant, so it remains to rule out the case in which all these possibilities occur.

 Say that an $x$-variable $x_{ij}$ has $t$-type $\tau \in \{0,D\}$ if the $x$-index $A$ in which its denotation appears satisfies $A_t = \tau$. By assumption, there are some variables $x_{ij},y_{pq},z_{rs}$ of $t$-type $D$. The product $x_{ip} y_{pq} z_{qi}$ corresponds to some index triple $(A,B,C) \in S$ satisfying $B_t = D$. We conclude that $(A_t,B_t,C_t) = (0,D,0)$ and so $x_{ip}$ has $t$-type $0$. Similarly, the product $x_{sp} y_{pr} z_{rs}$ shows that $y_{pr}$ has $t$-type $0$, and the product $x_{ij} y_{jr} z_{ri}$ shows that $z_{ri}$ has $t$-type $0$. But then the product $x_{ip} y_{pr} z_{ri}$ corresponds to some index triple $(A,B,C) \in S$ satisfying $(A_t,B_t,C_t) = (0,0,0)$, which is impossible. This contradiction completes the proof.
\end{proof}

\subsection{Upper bound for the Coppersmith--Winograd tensor} \label{sec:upper-bound-special}

Armed with Lemma~\ref{lem:coherent-sums}, we can prove the upper bound on the merging value. In order to simplify the notations involved, we first prove the upper bound in the special case in which the tensor being analyzed is the Coppersmith--Winograd tensor $\tcw$. In the following subsection we generalize the argument to arbitrary Coppersmith--Winograd-like tensors.

\begin{theorem} \label{thm:upper-bound-special}
 For every $q \geq 2$ and any $\rho\in[2,3]$,
\[ \log \Vm_\rho(\tcwq) \leq \max_{\alpha \in [0,1]} H(\tfrac{2-\alpha}{3},\tfrac{2\alpha}{3},\tfrac{1-\alpha}{3}) + \tfrac{\alpha\rho}{3} \log q + \tfrac{\rho-2}{3} H(\tfrac{1-\alpha}{2},\alpha,\tfrac{1-\alpha}{2}). \]
\end{theorem}
\begin{proof}
As usual we write $\tcw$ for $\tcwq$.
 Given an integer $N$, we will upper bound $\Vm_{\rho,N}(\tcw)$. Let $S$ be a consistent restriction of $\tcw^{\otimes N}$ and let $S = \sum_i S_i$ be its decomposition into disjoint coherent sums such that $S_i \approx \mm{n_i,m_i,p_i}$ and $\Vm_{\rho,N}(\tcw) = \sum_i (n_im_ip_i)^{\rho/3}$. We call each $S_i$ a \emph{line}.
 Lemma~\ref{lem:coherent-sums} shows that for each line $S_i$, each $t \in [N]$ is either $x$-constant, $y$-constant or $z$-constant. If there are $\gamma_x N,\gamma_y N,\gamma_z N$ of each then we say that $S_i$ has \emph{line type} $\tau_\ell = (\gamma_x,\gamma_y,\gamma_z)$. Since $\gamma_x,\gamma_y,\gamma_z$ are necessarily multiples of $1/N$, and $\gamma_x+\gamma_y+\gamma_z=1$, there are $O(N^2)$ different line types.

 Each $S_i$ results from merging several constituent tensors of $\tcw^{\otimes N}$. We will sometimes think of $S_i$ as the set of these constituent tensors.
 A constituent tensor $T$ of $\tcw^{\otimes N}$ which results from multiplying $\alpha_xN,\alpha_yN,\alpha_zN,\beta_xN,\beta_yN,\beta_zN$ each of the constituent tensors $\tcw_{0,1,1},\tcw_{1,0,1},\tcw_{1,1,0},\tcw_{2,0,0},\tcw_{0,2,0}$, $\tcw_{0,0,2}$ of $\tcw$ (respectively) is said to have \emph{type} $(\alpha_x,\alpha_y,\alpha_z,\beta_x,\beta_y,\beta_z)$. Since these numbers are multiples of $1/N$ with sum 1,
there are $O(N^5)$ possible types. We let $\Vol_\tau(S_i)$ be the sum of the volumes of all $T \in S_i$ of type $\tau$. Since the volume is the number of basic products $xyz$, it follows that $\Vol(S_i) = \sum_\tau \Vol_\tau(S_i)$.

 Consider a specific line type $\tau_\ell = (\gamma_x,\gamma_y,\gamma_z)$ and a specific type $\tau = (\alpha_x,\alpha_y,\alpha_z,\beta_x,\beta_y,\beta_z)$. We will upper bound
\[ U_{\rho,N}(\tau_\ell,\tau) = \sum_{i\colon S_i \text{ has line type } \tau_\ell} \Vol_\tau(S_i)^{\rho/3}. \]
 This implies an upper bound on $\Vm_{\rho,N}(\tcw)$ as follows. First, $\rho \leq 3$ implies that $(\alpha+\beta)^{\rho/3} \leq \alpha^{\rho/3} + \beta^{\rho/3}$ (this follows from Minkowski's inequality, for example), and so
\begin{align*}
 \sum_{i\colon S_i \text{ has line type } \tau_\ell} \Vol(S_i)^{\rho/3} &=
 \sum_{i\colon S_i \text{ has line type } \tau_\ell} \left(\sum_\tau \Vol_\tau(S_i)\right)^{\rho/3} \\  &\leq
 \sum_{i\colon S_i \text{ has line type } \tau_\ell} \sum_\tau \Vol_\tau(S_i)^{\rho/3} \\ &=
 \sum_\tau U_{\rho,N}(\tau_\ell,\tau) \\ &\leq O(N^5) \max_\tau U_{\rho,N}(\tau_\ell,\tau).
\end{align*}
 Summing over all $\tau_\ell$,
\[
 \Vm_{\rho,N}(\tcw) \leq O(N^7) \max_{\tau_\ell,\tau} U_{\rho,N}(\tau_\ell,\tau).
\]
 When taking the $N$th root and letting $N\to\infty$, the factor $O(N^7)$ disappears. Therefore
\begin{equation} \label{eq:crucial}
 \Vm(\tcw) \leq \max_{\tau_\ell,\tau} \lim_{N\to\infty} U_{\rho,N}(\tau_\ell,\tau)^{1/N}.
\end{equation}

 Let $\alpha = \alpha_x+\alpha_y+\alpha_z$ and $\beta = \beta_x+\beta_y+\beta_z$, and define
\begin{align*}
 P_x &= \exp_2 H(\alpha_x + \beta_y + \beta_z, \alpha_y + \alpha_z, \beta_x)N, \\
 P_y &= \exp_2 H(\beta_x + \alpha_y + \beta_z, \alpha_x + \alpha_z, \beta_y)N, \\
 P_z &= \exp_2 H(\beta_x + \beta_y + \alpha_z, \alpha_x + \alpha_y, \beta_z)N, \\
 Q_x &= \exp_2 H\big(\tfrac{\alpha_x + \beta_y + \beta_z - \gamma_x}{1-\gamma_x}, \tfrac{\alpha_y + \alpha_z}{1-\gamma_x}, \tfrac{\beta_x}{1-\gamma_x}\big)(1-\gamma_x)N, \\
 Q_y &= \exp_2 H\big(\tfrac{\beta_x + \alpha_y + \beta_z - \gamma_y}{1-\gamma_y}, \tfrac{\alpha_x + \alpha_z}{1-\gamma_y}, \tfrac{\beta_y}{1-\gamma_y}\big)(1-\gamma_y)N, \\
 Q_z &= \exp_2 H\big(\tfrac{\beta_x + \beta_y + \alpha_z - \gamma_z}{1-\gamma_z}, \tfrac{\alpha_x + \alpha_y}{1-\gamma_z}, \tfrac{\beta_z}{1-\gamma_z}\big)(1-\gamma_z)N.
\end{align*}
 Here $P_x,P_y,P_z$ are upper bounds on the number of different $x,y,z$-indices, respectively. The quantities $Q_x,Q_y,Q_z$ are upper bounds on the number of different $x,y,z$--indices, respectively, that can appear in any given line. The reason that $Q_x$ bounds the number of $x$-indices is that a $\gamma_x$-fraction of the indices are fixed at $0$, and these have to be deducted from the $\alpha_x+\beta_y+\beta_z$-fraction which is $0$ among the entire $N$ coordinates. The resulting distribution then applies only to the remaining $(1 - \gamma_x)N$ coordinates.

 From now on, we consider only lines of line type $\tau_\ell$. Let $I_t,J_t,K_t$ be the number of $x,y,z$-indices, respectively, in tensors of type $\tau$ in line $S_t$. Note that $\sum_t I_t \leq P_x$, $\sum_t J_t \leq P_y$, $\sum_t K_t \leq P_z$. As noted above, $I_t \leq Q_x$, $J_t \leq Q_y$, $K_t \leq Q_z$. In order to upper bound $\Vol_\tau(S_t)$, notice that if a matrix multiplication tensor involves $X,Y,Z$ each of $x,y,z$-variables, respectively, then its volume is $\sqrt{XYZ}$: indeed, for $\mm{n,m,p}$ we have $X=nm$, $Y=mp$, $Z=pn$ and the volume is $\Vol(\mm{n,m,p}) = nmp = \sqrt{XYZ}$.
 Each $x$-index contains exactly $(\alpha_y+\alpha_z)N$ coordinates equal to~$1$, and so it corresponds to $q^{(\alpha_y+\alpha_z)N}$ variables. Therefore
\[ \Vol_\tau(S_t) = \sqrt{q^{(\alpha_y+\alpha_z)N} I_t q^{(\alpha_x+\alpha_z)N} J_t q^{(\alpha_x+\alpha_y)N} K_t} = q^{\alpha N} \sqrt{I_t J_t K_t}. \]
 In total, we obtain the upper bound
\[
 U_{\rho,N}(\tau_\ell,\tau) \leq q^{(\alpha \rho/3) N} \sum_t (I_t J_t K_t)^{\rho/6}.
\]

 Let us focus now on the quantity
\[
 \qS = \sum_t (I_t J_t K_t)^{\rho/6}.
\]
 We want to obtain an upper bound on $\qS$.
 We can assume that $\sum_t I_t = P_x$, $\sum_t J_t = P_y$, $\sum_t K_t = P_z$. Lagrange multipliers show that this quantity is optimized when $I_t^{\rho/6-1} (J_t K_t)^{\rho/6},\allowbreak J_t^{\rho/6-1} (I_t K_t)^{\rho/6},\allowbreak K_t^{\rho/6-1} (I_t J_t)^{\rho/6}$ are all constant. Multiplying all these constraints together, we get that $I_t J_t K_t$ is constant (assuming $\rho \neq 2$) and so $I_t,J_t,K_t$ are constant. In order to find the constants, let $\qT$ be the number of different summands. Then $I_t = P_x/\qT$, $J_t = P_y/\qT$, $K_t = P_z/\qT$. On the other hand, $I_t \leq Q_x$, $J_t \leq Q_y$, $K_t \leq Q_z$, and so $\qT \geq \max (P_x/Q_x,P_y/Q_y,P_z/Q_z) \geq \sqrt[3]{P_xP_yP_z/Q_xQ_yQ_z}$. Therefore
\[
 \qS \leq \max_{\qT \geq \sqrt[3]{P_xP_yP_z/Q_xQ_yQ_z}} \qT^{1-\rho/2} (P_x P_y P_z)^{\rho/6}.
\]
 Since $1-\rho/2 \leq 0$, we would like $\qT$ to be as small as possible, and so
\[
 \qS \leq (P_xP_yP_z/Q_xQ_yQ_z)^{1/3-\rho/6} (P_xP_yP_z)^{\rho/6} = (P_xP_yP_z)^{1/3} (Q_xQ_yQ_z)^{(\rho-2)/6}.
\]
 Altogether, we obtain the upper bound
\[
 U_{\rho,N}(\tau_\ell,\tau) \leq q^{(\alpha \rho/3)N} (P_xP_yP_z)^{1/3} (Q_xQ_yQ_z)^{(\rho-2)/6}.
\]

The concavity of the entropy function shows that
\begin{align*}
 \frac{1}{N} \log &(P_xP_yP_z)^{1/3} \\
 =&\frac{H(\alpha_x\! +\! \beta_y \!+\! \beta_z, \alpha_y \!+ \!\alpha_z, \beta_x) + H(\beta_x \!+\! \alpha_y \!+v \beta_z, \alpha_x \!+\! \alpha_z, \beta_y) + H(\beta_x \!+\!\beta_y \!+\! \alpha_z, \alpha_x \!+\! \alpha_y, \beta_z)}{3} \\ \leq
 &H(\tfrac{\alpha+2\beta}{3},\tfrac{2\alpha}{3},\tfrac{\beta}{3}) = H(\tfrac{2-\alpha}{3},\tfrac{2\alpha}{3},\tfrac{1-\alpha}{3}).
\end{align*}
 Similarly,
\begin{align*}
 \frac{1}{N} \log (Q_xQ_yQ_z)^{1/2}
=& \tfrac{1-\gamma_x}{2} H(\tfrac{\alpha_x + \beta_y + \beta_z - \gamma_x}{1-\gamma_x}, \tfrac{\alpha_y + \alpha_z}{1-\gamma_x}, \tfrac{\beta_x}{1-\gamma_x}) +
 \tfrac{1-\gamma_y}{2} H(\tfrac{\beta_x + \alpha_y + \beta_z - \gamma_y}{1-\gamma_y}, \tfrac{\alpha_x + \alpha_z}{1-\gamma_y}, \tfrac{\beta_y}{1-\gamma_y}) +\\
 &\hspace{65mm}
 \tfrac{1-\gamma_z}{2} H(\tfrac{\beta_x + \beta_y + \alpha_z - \gamma_z}{1-\gamma_z}, \tfrac{\alpha_x + \alpha_y}{1-\gamma_z}, \tfrac{\beta_z}{1-\gamma_z}) \\
 \leq &H(\tfrac{\alpha+2\beta-1}{2},\alpha,\tfrac{\beta}{2}) = H(\tfrac{1-\alpha}{2},\alpha,\tfrac{1-\alpha}{2}).
\end{align*}

 Therefore
\[
 U_{\rho,N}(\tau_\ell,\tau)^{1/N} \leq q^{\alpha\rho/3} \exp_2 H(\tfrac{2-\alpha}{3},\tfrac{2\alpha}{3},\tfrac{1-\alpha}{3}) \exp_2 [H(\tfrac{1-\alpha}{2},\alpha,\tfrac{1-\alpha}{2})\tfrac{\rho-2}{3}].
\]
 The theorem now follows from~\eqref{eq:crucial}.
\end{proof}

\subsection{Upper bound for Coppersmith--Winograd-like tensors} \label{sec:upper-bound-general}
Extending the proof of Theorem~\ref{thm:upper-bound-special} to general Coppersmith--Winograd-like tensors involves mainly notational difficulties. Before stating the theorem, we need to describe the general form of the penalty term, that is, the last summand in the theorem.

Let $T$ be a symmetric partitioned tensor.
The proof will include an upper bound on all distributions in $\dist(T)$, which correspond to the types appearing in the proof of Theorem~\ref{thm:upper-bound-special}. The proof of the theorem shows that the worst bound is obtained on symmetric distributions, and accordingly we concentrate on these.
For any $P\in \sdist(T)$, define
\[
P_0=\sum_{s\in\suppz(T)}P(s),
\]
and define
the function $P^\ast\colon \supp(T)\to [0,1]$ as follows:
\[
P^\ast(s)=
\left\{
\begin{array}{ll}
P(s) &\textrm{ if }s\in \suppn(T),\\
0&\textrm{ otherwise}.
\end{array}
\right.
\]
Let $P_{\textrm{m}}\colon \Int \to [0,1]$ denote the marginal function of $P$:
for any $i\in \Int$,
\[
P_{\textrm{m}}(i)=\sum_{(i,j,k)\in \supp(T)}P(i,j,k).
\]
Note that $P_{\textrm{m}}$ is a probability distribution.
Similarly, let $P^\ast_{\textrm{m}}\colon \Int \to [0,1]$ denote the marginal function of~$P^\ast$.
Define the probability distribution $\tilde{P}\colon \Int\to [0,1]$ as follows:
\[
\tilde P(i) =
\left\{
\begin{array}{ll}
\frac{3}{2P_0}(P_{\textrm{m}}(0)-\frac{P_0}{3}) &\textrm{ if }i=0,\\
\frac{3}{2P_0}(P_{\textrm{m}}(i)-P^\ast_{\textrm{m}}(i))&\textrm{ otherwise}.
\end{array}
\right.
\]
Note that this is indeed a probability distribution, since
\[
\sum_{i\in\Int}\tilde P(i)=\frac{3}{2P_0}\left(1-\frac{P_0}{3}-(1-P_0)\right)=1.
\]

We can now state the general upper bound.

\begin{theorem}\label{thm:upper-bound}
Let $T$ be a Coppersmith--Winograd-like estimated symmetric partitioned tensor.
For any $\rho\in[2,3]$, the merging value $\Vm_\rho(T)$ is upper bounded by
\begin{equation*}
\log \Vm_\rho(T) \leq \max_{P\in\sdist(T)}
H(P_{\mathrm{m}})+\sum_{s\in \supp(T)} P(s)\log (\Val_\rho(T_s))+\frac{\rho-2}{3}\times P_0\times H(\tilde{P}).
\end{equation*}
\end{theorem}

\begin{proof}[\Proofsketch]

Given an integer $N$, we will upper bound $\Vm_{\rho,N}(T)$. Let $S$ be a consistent restriction of $T^{\otimes N}$, and let $S = \sum_i S_i$ be its decomposition into disjoint coherent sums, so that $\Vm_{\rho,N}(T) = \sum_i \Val_\rho(S_i)$. We call each $S_i$ a \emph{line}. Lemma~\ref{lem:coherent-sums} shows that for each line $S_i$, each $t \in [N]$ is either $x$-constant, $y$-constant, $z$-constant, or the $t$th coordinate of all summands in $S_i$ is some fixed $s \in \suppn(T)$.
As in the proof of Theorem~\ref{thm:upper-bound-special}, it is enough to bound $\sum_i \Vol_P(\tz{S_i})^{\rho/3} \Val_\rho(\tn{S_i})$ for all distributions $P \in \dist(T)$.
Calculation shows that the worst distribution is symmetric, so fix some $P \in \sdist(T)$.
Further calculation shows that the largest contribution to $\sum_i \Val_\rho(S_i)$ comes from lines in which a $P_0/3$~fraction of the coordinates are $x$-constant, $y$-constant, and $z$-constant each, and a $1-P_0$~fraction of the coordinates are fixed. We call such a line a \emph{typical line}. So our goal is to bound
\[ V = \sum_{i\colon S_i\text{ is typical}} \Vol_P(\tz{S_i})^{\rho/3} \Val_\rho(\tn{S_i}). \]

The total number $R$ of $x$-indices in typical lines is easily seen to satisfy $\log R \approx NH(\Pm)$.
Now let $S_i$ be a typical line. The number $Q$ of $x$-indices of summands in $S_i$ satisfies $\log Q \approx \frac{2P_0}{3} N H(\tilde{P})$. Indeed, only the $y$-constant and $z$-constant coordinates are not fixed, and there are $\frac{2P_0}{3} N$ of them. Among the fixed coordinates, the $x$-constant contain $\frac{P_0}{3} N$ zeroes, and the others contain $\Pm^\ast(i)N$ coordinates whose value is $i$. It is not hard to check that the distribution of the $y$-constant and $z$-constant coordinates is indeed given by $\tilde{P}$.

As the proof of Theorem~\ref{thm:upper-bound-special} shows, we can upper bound $V$ by assuming that each typical line contains the maximal number of $x$-indices, $y$-indices, and $z$-indices, namely $Q$. The number of typical lines is thus $R/Q$.

We proceed to calculate $\Vol(\tz{S_i})$. Recall that $\tz{S_i} \approx \sum_{t \in S_i} \tz{t}$, and for each $t \in S_i$, the volume $U = \Vol(\tz{t})$ is $U = \prod_{s \in \suppz(T)} \Vol(T_s)^{NP(s)}$. If $\tz{t} \approx \mm{n,m,p}$ then each $x$-index corresponds to $nm$ $x$-variables, each $y$-index to $mp$ $y$-variables, and each $z$-index to $pn$ $z$-variables. The tensor $\tz{S_i}$ has $X = Qnm$ $x$-variables, $Y = Qmp$ $y$-variables, and $Z = Qpn$ $z$-variables, so its volume is $\Vol(\tz{S_i}) = \sqrt{XYZ} = Q^{3/2} nmp = Q^{3/2} U$.
Therefore
\[
V \leq \frac{R}{Q} \Vol(\tz{S_i})^{\rho/3} \Val_\rho(\tn{S_i}) = R Q^{\rho/2-1} U^{\rho/3} \prod_{s \in \suppn(T)} \Val_\rho(T_s)^{NP(s)}.
\]
Taking the logarithm, we deduce
\begin{align*}
\log V &\lesssim NH(\Pm) + \frac{\rho-2}{3} \times NP_0 \times H(\tilde{P}) + \frac{\rho}{3} \sum_{s \in \suppz(T)} NP(s)\Vol(T_s) + \sum_{s \in \suppn(T)} NP(s) \Val_\rho(T_s) \\ &=
NH(\Pm) + \frac{\rho-2}{3} \times NP_0 \times H(\tilde{P}) + \sum_{s \in \supp(T)} NP(s)\Val_\rho(T_s).
\end{align*}
This shows that $\log V^{1/N} = (\log V)/N$ is upper bounded in the limit by the expression given by the theorem.
\end{proof}

\subsection{Numerical calculations} \label{sec:simulations}
We now analyze the canonical recursive Coppersmith--Winograd approach by applying Theorem~\ref{thm:upper-bound} to the canonical $\tcw^{\otimes 2^r}$ for several values $r\ge 0$.
Since, as described in footnote~\ref{fn:exact-calc} on page~\pageref{fn:exact-calc}, Theorem~\ref{thm:cw-formula} generally does not give an exact formula for computing the \pr value of each constituent tensor, we use the upper bound appearing in Theorem~\ref{thm:cw-formula} for estimating these values in our calculations (instead of the lower bound of Theorem~\ref{thm:cw-formula}, as was done in Section~\ref{sec:cw-powers}); this can only deteriorate the results we obtain.

The numerical results of this analysis\footnote{All the programs used to perform the numerical calculations
described in this subsection
are available as
\url{http://www.francoislegall.com/MatrixMultiplication/programsLB.zip}.}  are given in Table~\ref{tab:simulations}, and can be interpreted as follows:
for given $r$ and $q$, the corresponding value presented in the table is the solution $\alpha$ of $\Vmub_\alpha(\tcw^{\otimes 2^r}) = (q+2)^{2^r}$, where $\Vmub_\alpha(\tcw^{\otimes 2^r})$ is the upper bound on $\Vm_\alpha(\tcw^{\otimes 2^r})$ given by Theorem~\ref{thm:upper-bound}, and is the best value that can be possibly obtained from the canonical $\tcw^{\otimes 2^r}$ and its powers.
In particular, this shows that analyzing the
64th, 128th, 256th powers and higher powers of the tensor $\tcw$ for $q=5$ using the the canonical recursive Coppersmith--Winograd approach cannot give an upper bound on~$\omega$ smaller than $2.3725$.

\begin{table}[h!]
\renewcommand\arraystretch{1}
\begin{center}
\caption{The solution $\alpha$ of the equation $\Vmub_\alpha(\tcw^{\otimes 2^r}) = (q+2)^{2^r}$, rounded down to five decimal digits, for several values of $r$ and $q$.}\label{tab:simulations}\vspace{3mm}
\begin{tabular}{c|c|c|c|c|c|c|}
&$r=0$&$r=1$&$r=2$&$r=3$&$r=4$\\
\hline
$q=1$&$2.2387$&$2.3075$&$2.4587$&$2.5772$&$2.6184$\\
$q=2$&$2.2540$&$2.3181$&$2.4187$&$2.4623$&$2.4673$\\
$q=3$&$2.2725$&$2.3203$&$2.3834$&$2.4015$&$2.4025$\\
$q=4$&$2.2907$&$2.3262$&$2.3690$&$2.3788$&$2.3791$\\
$q=5$&$2.3078$&$2.3349$&$2.3659$&$2.3723$&$2.3725$\\
$q=6$&$2.3234$&$2.3448$&$2.3682$&$2.3731$&$2.3733$\\
$q=7$&$2.3377$&$2.3550$&$2.3733$&$2.3775$&$2.3776$\\
$q=8$&$2.3508$&$2.3651$&$2.3798$&$2.3833$&$2.3834$\\
\hline
\end{tabular}
\end{center}
\end{table}

\section{Generalizations} \label{sec:generalizations}
In this section we show how our results can be extended to analyze the limitations of current implementations of the laser method applied to a large class of partitioned tensors, much larger than the class of Coppersmith--Winograd-like tensors.

We start by giving another definition of merging, which we call $\emph{coherent merging}$.
\begin{definition} \label{def:merging-general}
 Let $T$ be a symmetric estimated partitioned tensor. Given $\rho \in [2,3]$ and $N \geq 1$, we define $\Vc_{\rho,N}(T)$ to be the maximum of $\sum_{s \in \supp(T')} \Val_\rho(T'_s)$ over all consistent restrictions $T'$ of $T^{\otimes N}$ in which each $\tz{{T'_s}}$ is a coherent sum.
 The \emph{coherent merging value} of $T$ is the function
\[ \Vc_\rho(T) = \lim_{N\to\infty} \Vc_{\rho,N}(T)^{1/N}. \]
\end{definition}
The only difference with Definition \ref{def:merging} is that we now require that all consistent sums be coherent.
Naturally, for any symmetric estimated partitioned tensor $T$ and any $\rho \in [2,3]$, we have
\[
\Vm_\rho(T) \ge \Vc_\rho(T) \ge \Vcw_\rho(T).
\]
Note that Lemma \ref{lem:coherent-sums} implies that $\Vm_\rho(T) = \Vc_\rho(T)$ for any Coppersmith--Winograd-like tensor~$T$.
The reason for introducing this new version of merging is that Lemma \ref{lem:coherent-sums} may not hold for arbitrary tensors: there exist (symmetric) partitioned tensors for which consistent but not coherent sums of zero-sequences of constituent tensors exist. Since all known implementations of the laser method nevertheless construct coherent restrictions of $T^{\otimes N}$, and in particular no general technique is known for constructing restrictions with consistent but not coherent sums of zero-sequences,
coherent merging represents what can be done by current implementations of the laser method.

The class of tensors to which the techniques developed in this paper can be applied
is defined as follows.
\begin{definition}
A symmetric partitioned tensor $T$ with support $\supp(T)\subseteq \Int\times\Int\times \Int$ belongs to the class $\classtensor$ if
\[
\suppz(T)=\supp(T)\cap\Big((\{0\}\times\Int\times\Int) \cup (\Int\times\{0\}\times\Int) \cup (\Int\times\Int\times\{0\})\Big),
\]
where $\suppz(T)\subseteq \supp(T)$ is the set defined in Definition \ref{def:consistent-sum}.
\end{definition}
Note that, while the class $\classtensor$ includes all Coppersmith--Winograd-like tensors, it is much larger.
It includes in particular tensors that are not tight.

We can now state our most general result.
\begin{theorem}\label{thm:upper-bound-general}
Let $T$ be an estimated tensor in $\classtensor$.
For any $\rho\in[2,3]$, the coherent merging value $\Vc_\rho(T)$ is upper bounded by
\begin{equation*}
\log \Vc_\rho(T) \leq \max_{P\in\sdist(T)}
H(\Pm)+\sum_{s\in \supp(T)} P(s)\log (\Val_\rho(T_s))+\frac{\rho-2}{3}\times P_0\times H(\tilde{P}).
\end{equation*}
\end{theorem}
The proof of Theorem \ref{thm:upper-bound-general} is exactly the same as the proof of Theorem \ref{thm:upper-bound}, since the only properties of the estimated symmetric partitioned tensor $T$ actually used in the proof of Theorem~\ref{thm:upper-bound} were that $T\in \classtensor$ and the fact that only coherent sums need be considered due to Lemma \ref{lem:coherent-sums}. Indeed, in our case the latter property trivially holds from the definition of the coherent merging value.

\section{Discussion} \label{sec:discussion}

Our main result shows that the conjecture $\omega = 2$ cannot be proved using the laser method with merging applied to the tensor $\tcw$. On the other hand, we believe that the technique can be used to improve known bounds on $\omega$. We believe that it is possible that
\[ \Vm_\rho(\tcw) > \limsup_{r\to\infty} \Vcw_\rho(\tcw^{\otimes 2^r}). \]
The reason is that $\Vcw_{\rho,N/2^r}(\tcw^{\otimes 2^r})$ corresponds to a lower bound on $\Vm_{\rho,N}(\tcw)$ in which merging is done in groups of $2^r$ coordinates at a time, for fixed $r$; if the merging width $2^r$ is allowed to vary with~$N$, then a better lower bound on $\Vm_{\rho,N}(\tcw)$ can potentially be obtained.

\medskip

Our main result gives a limit on the possible upper bounds on $\omega$ obtainable for given $q \geq 1$ which deteriorates as $q$ gets smaller. In contrast, for known constructions the best $q$ is $q = 5$ (or $q = 6$ for the construction without merging), a behavior which is also apparent in the upper bounds we get for $\tcw^{\otimes 4}$ and higher powers. This leads us to suspect that our upper bound on the merging value is not tight. We leave it as an open question to determine the correct value of $\Vm_\rho(\tcw)$.

A similar issue concerns the \pr value $\Vcw_\rho(\tcw^{\otimes 2^r})$. Theorem~\ref{thm:cw-formula} can be used to calculate the value for $r = 0$ and $r = 1$, but already for $r = 2$ there is a gap between the lower and upper bounds. We conjecture that the lower bound is tight, but have so far been unable to prove this. One concrete reason for this conjecture is that in certain cases, if we assume that the upper bound is tight then we can obtain a bound $\omega \leq \rho$ for some $\rho < 2$. A deeper reason is that the upper bound does not rely on the fact that the tensor $T'$ in the definition of $\Vcw_{\rho,N}(T)$ is obtained from $T^{\otimes N}$ by a partitioned restriction. Rather, it only depends on the fact that the constituent tensors of $T'$ are on disjoint variables. Indeed, if we do not insist that $T'$ be obtained by a partitioned restriction, then the upper bound is tight. The main difficulty in proving the lower bound is to construct $T'$ using a partitioned restriction; without this stipulation, a simple randomized construction matches the upper bound.

Our upper bound on $\Vm_\rho(\tcw)$ also ignores the fact that $T'$ is a partitioned restriction, and is also tight if do not insist that $T'$ be obtained by a partitioned restriction. This is another reason to believe that the upper bound on $\Vm_\rho(\tcw)$ is not tight.

\medskip

Research in matrix multiplication has proceeded in the past by finding new techniques and new identities (corresponding to upper bounds on ranks or border ranks of tensors). Notwithstanding recent developments, and ignoring the group-theoretic method which so far has not produced new upper bounds on $\omega$, this process seems to have stagnated. While the new technique we propose in this paper could potentially lead to improved bounds on $\omega$, we nevertheless find the most promising research direction (besides the group-theoretic method~\cite{CohnUmans03,CKSU} and the $s$-rank~\cite{CohnUmans13}) to be \emph{finding new identities}\footnote{A tantalizing source for new identities is the ``basic'' Coppersmith--Winograd identity itself and its powers, obtained by setting $x^{[2]}_{q+1} = y^{[2]}_{q+1} = z^{[2]}_{q+1} = 0$. As discussed in~\cite{CoppersmithWinograd}, it is possible that the $N$th tensor power of this tensor has border rank significantly lower than the known upper bound $(q+2)^N$, though so far no new bounds are known for any $N$.}. Perhaps a systematic search for new identities could be automated and would lead to significantly improved upper bounds on $\omega$. The lower bound techniques we developed may be instrumental for such a search, since they make possible to immediately rule out unpromising identities, i.e., to show that a given identity, and \emph{any} of its powers, even accounting for any possible repartitioning scheme, cannot lead to $\omega=2$.


\appendix

\section{Proofs of results from Section~\ref{sec:value}} \label{sec:value-proofs}

We start by showing that the limit in the definition of $\Vcw_{\rho,P}(T)$ exists. If $\prnd{(N_1+N_2)}{P} = \prnd{N_1}{P} + \prnd{N_2}{P}$ then it is not hard to check that $\Vcw_{\rho,P,N_1+N_2}(T) \geq \Vcw_{\rho,P,N_1}(T) \Vcw_{\rho,P,N_2}(T)$. While this inequality is not true in general due to rounding, it is true approximately. A careful application of Fekete's lemma then shows that the limit exists.

We proceed to prove Theorem~\ref{thm:cw-formula-exact}.
\begin{proof}[Proof of Theorem~\ref{thm:cw-formula-exact}]
 It is not hard to check that $\Vcw_\rho(T) \geq \Vcw_{\rho,P}(T)$ for every $P \in \dist(T)$, and we proceed to prove the other direction.

 For $P \in \dist(T),P' \in \dist(\rot{T}),P'' \in \dist(\rott{T})$, define $\Vcw_{\rho,P,P',P'',N}(T)$ by naturally extending the definition of $\Vcw_{\rho,P,N}(T)$ to allow different distributions for factors coming from $T,\rot{T},\rott{T}$.

 Given $N$, notice that each $s \in \supp((T\otimes\rot{T}\otimes\rott{T})^{\otimes N})$ corresponds to some distributions $P_s \in \dist(T),P'_s \in \dist(\rot{T}),P''_s \in \dist(\rott{T})$ (obtained by dividing the actual quantities by $N$), and there are $N^{O(1)}$ many such triples of distributions, forming a set $\dist_N$. It is not hard to check that
 \[ \Vcw_{\rho,N}(T) \leq \sum_{(P,P',P'') \in \dist_N} \Vcw_{\rho,P,P',P'',N}(T) \leq N^{O(1)} \max_{(P,P',P'') \in \dist_N} \Vcw_{\rho,P,P',P'',N}(T). \]
 For each $N$, let $(P_N,P'_N,P''_N) \in \dist_N$ be the distribution maximizing $\Vcw_{\rho,P,P',P'',N}(T)$. The triples $P_N,P'_N,P''_N$ have an accumulation point $P,P',P''$ which satisfies $\lim\inf_{N\to\infty} \Vcw_{\rho,N}(T)^{1/3N} \leq\! \Vcw_{\rho,P,P'\!,P''}$ (since $(N^{O(1)})^{1/3N} \to 1$), showing that
 \[ \Vcw_\rho(T) \leq \max_{(P,P',P'') \in \dist(T)\times\dist(\rot{T})\times\dist(\rott{T})} \Vcw_{\rho,P,P',P''}(T). \]

 In order to complete the proof, we need to show that the maximum is obtained when $P' = \rot{P}$ and $P'' = \rott{P}$. For $P \in \dist(T),P' \in \dist(\rot{T}), P'' \in \dist(\rott{T})$, define $Q = (P+\rott{P'}+\rot{P''})/3$.
 Consider the partitioned degeneration $T'$ of $(T\otimes\rot{T}\otimes\rott{T})^{\otimes N}$ witnessing $\Vcw_{\rho,P,P',P'',N}(T)$. We can view $T' \otimes \rot{T'} \otimes \rott{T'}$ as a partitioned degeneration of $(T\otimes\rot{T}\otimes\rott{T})^{\otimes 3N}$ witnessing $\Vcw_{\rho,Q,3N}(T) \geq \Vcw_{\rho,P,P',P'',N}(T)^3$, and so $\Vcw_{\rho,Q}(T) \geq \Vcw_{\rho,P,\rot{P},\rott{P}}(T)$.
\end{proof}

Theorem~\ref{thm:cw-method} is a simple corollary.
\begin{proof}[Proof of Theorem~\ref{thm:cw-method}]
 Let $\rho \in [2,3]$, and suppose that $\Val_\rho(T_s) \leq \Vg_\rho(T_s)$ for all $s \in \supp(S)$. Let $P \in \dist(T)$ be a distribution such that $\Vcw_\rho(T) = \Vcw_{\rho,P}(T)$, which exists by Theorem~\ref{thm:cw-formula-exact}. Fix a value of $N$, and let $T'$ be a partitioned degeneration of $(T\otimes\rot{T}\otimes\rott{T})^{\otimes N}$ with strongly disjoint support such that $\Vcw_{\rho,P,N}(T) = \sum_{s \in \supp_P(T')} \Val_\rho(T'_s)$.
 Lemma~\ref{lem:value-properties} implies that $\Vg_\rho((T\otimes\rot{T}\otimes\rott{T})^{\otimes N}) \geq \Vcw_{\rho,P,3N}(T)$, and so $\Vg_\rho(T) \geq \Vcw_{\rho,P,3N}(T)^{1/3N}$. Taking the limit $N\to\infty$, we conclude that $\Vg_\rho(T) \geq \Vcw_{\rho,P}(T) = \Vcw_\rho(T)$.
\end{proof}

Finally, we prove the upper bound part of Theorem~\ref{thm:cw-formula}.
\begin{proof}[Proof of upper bound part of Theorem~\ref{thm:cw-formula}]
 Consider first general (not necessarily symmetric) tensors $T$.
 In view of Theorem~\ref{thm:cw-formula-exact}, it is enough to prove that for each $P \in \dist(T)$,
\[
 \log\Vcw_{\rho,P}(T) \leq \sum_{\ell=1}^3 \frac{H(P_\ell)}{3} + \EE_{s\sim P}[\log\Val_\rho(T_s)].
\]
 For any $N \geq 1$, $\Vcw_{\rho,P,3N} = \sum_{s \in \supp_P(T')} \Val_\rho(T'_s)$ for some partitioned restriction $T'$ of $(T\otimes\rot{T}\otimes\rott{T})^{\otimes N}$ with strongly disjoint support. For every $s\in\supp_P(T')$ we have $\log\Val_\rho(T^{\otimes N}_s) = 3N\EE_{\sigma\sim P}[\log\Val_\rho(T_\sigma)] \pm O(1)$. Since all $x$-indices in $\supp_P(T')$ are disjoint, $\log |\supp_P(T')| \leq NH(P_1)+NH(P_2)+NH(P_3)$, using the upper bound on the corresponding multinomial coefficient. 
 We conclude that
\[
 \log\Vcw_{\rho,P,N}(T) \leq 3N\sum_{\ell=1}^3 \frac{H(P_\ell)}{3} + 3N\EE_{s\sim P}[\log\Val_\rho(T_s)] + O(1).
\]
 The desired inequality follows by taking the limit $N\to\infty$.

 Suppose now that $T$ is symmetric, and consider any $P \in \dist(T)$. Let $Q = \frac{P+\rot{P}+\rott{P}}{3} \in \sdist(T)$. It is not hard to check that $\EE_{s\sim Q}[\log\Val_\rho(T_s)] = \EE_{s\sim P}[\log\Val_\rho(T_s)]$, and concavity of the entropy function shows that $\frac{H(P_1)+H(P_2)+H(P_3)}{3} \leq H(Q_1) = \frac{H(Q_1)+H(Q_2)+H(Q_3)}{3}$. This shows that there is a symmetric distribution maximizing the upper bound.
\end{proof}
\fi

\end{document}